\documentclass[letterpaper, 11pt]{article}
\usepackage{fullpage}
\usepackage{graphicx}
\usepackage{multirow}
\usepackage{amsmath, amsthm, amssymb, cite, color, comment, enumitem, mathtools, subcaption, url}
\usepackage[unicode]{hyperref}
\usepackage[subrefformat=parens]{subcaption}
\usepackage{autonum}
\numberwithin{equation}{section}

\usepackage{algorithm}
\usepackage[noend]{algorithmic}

\newtheorem{theorem}{Theorem}[section]
\newtheorem{lemma}[theorem]{Lemma}
\newtheorem{proposition}[theorem]{Proposition}

\newtheorem{claim}[theorem]{Claim}
\newtheorem{corollary}[theorem]{Corollary}

\theoremstyle{definition}
\newtheorem{definition}[theorem]{Definition}

\newtheorem*{problem}{Problem}

\theoremstyle{remark}

\newcommand{\rmdp}{{\mathrm{dp}}}
\newcommand{\twdp}{{\mathrm{twdp}}}
\newcommand{\auxdp}{\omega}
\newcommand{\lpath}{\lambda}

\newcommand{\argmin}{\mathop{\rm arg\,min}\limits}
\newcommand{\Opt}{{\mathrm{Opt}}}
\newcommand{\Feas}{{\mathrm{Feas}}}
\newcommand{\IG}{{\mathrm{IG}}}
\newcommand{\cH}{{\mathcal{H}}}
\newcommand{\Vdd}{{V_{{}^\bullet\!{}_\bullet}}}
\newcommand{\tw}{{\mathrm{tw}}}
\newcommand{\GMMN}{P}

\title{Dynamic Programming Approach to~the~Generalized~Minimum~Manhattan~Network~Problem\thanks{A preliminary version \cite{MOY} of this paper will appear in ISCO 2020.}}
\author{Yuya Masumura\thanks{Fast Retailing Co., Ltd., Tokyo 135-0063, Japan. Email: \texttt{yuya.masumura@fastretailing.com}}  \and
Taihei Oki\thanks{The University of Tokyo, Tokyo 113-8656, Japan. Email: \texttt{taihei\_oki@mist.i.u-tokyo.ac.jp}} \and
Yutaro Yamaguchi\thanks{Kyushu University, Fukuoka 819-0395, Japan. Email: \texttt{yutaro\_yamaguchi@inf.kyushu-u.ac.jp}}}
\date{\empty}

\begin{document}
\maketitle
\thispagestyle{empty}

\begin{abstract}
We study the generalized minimum Manhattan network (GMMN) problem: given a set $\GMMN$ of pairs of points in the Euclidean plane $\mathbb R^2$, we are required to find a minimum-length geometric network which consists of axis-aligned segments and contains a shortest path in the $L_1$ metric (a so-called Manhattan path) for each pair in $\GMMN$.
This problem commonly generalizes several NP-hard network design problems that admit constant-factor approximation algorithms, such as the rectilinear Steiner arborescence (RSA) problem, and it is open whether so does the GMMN problem.

As a bottom-up exploration, Schnizler (2015) focused on the intersection graphs of the rectangles defined by the pairs in $\GMMN$, and gave a polynomial-time dynamic programming algorithm for the GMMN problem whose input is restricted so that both the treewidth and the maximum degree of its intersection graph are bounded by constants.
In this paper, as the first attempt to remove the degree bound, we provide a polynomial-time algorithm for the star case, and extend it to the general tree case based on an improved dynamic programming approach.
\end{abstract}


\clearpage
\thispagestyle{empty}
\setcounter{tocdepth}{2}
\tableofcontents
\clearpage
\setcounter{page}{1}

\section{Introduction}
In this paper, we study a geometric network design problem in the Euclidean plane $\mathbb{R}^2$.
For a pair of points $s$ and $t$ in the plane, a path between $s$ and $t$ is called a {\it Manhattan path} (or an {\it M-path} for short) if it consists of axis-aligned segments whose total length is equal to the Manhattan distance of $s$ and $t$ (in other words, it is a shortest $s$--$t$ path in the $L_1$ metric).
The {\it minimum Manhattan network (MMN) problem} is to find a minimum-length geometric network that contains an M-path for every pair of points in a given terminal set.
In the {\it generalized minimum Manhattan network (GMMN) problem}, given a set $\GMMN$ of pairs of terminals, we are required to find a minimum-length network that contains an M-path for every pair in $\GMMN$.
Throughout this paper, let $n = |\GMMN|$ denote the number of terminal pairs.

The GMMN problem was introduced by Chepoi, Nouioua, and Vax\`es~\cite{chepoi2008rounding}, and is known to be NP-hard as so is the MMN problem~\cite{chin2011minimum}.
The MMN problem and another NP-hard special case named the \emph{rectilinear Steiner arborescence (RSA) probelm} admit polynomial-time constant-factor approximation algorithms, and in~\cite{chepoi2008rounding} they posed a question whether so does the GMMN problem or not, which is still open.

Das, Fleszar, Kobourov, Spoerhase, Veeramoni, and Wolff~\cite{10.1007/s00453-017-0298-0} gave an $O(\log^{d+1} n)$-approximation algorithm for the $d$-dimensional GMMN problem based on a divide-and-conquer approach.
They also improved the approximation ratio for $d = 2$ to $O(\log n)$.
Funke and Seybold~\cite{funke2014generalized} (see also \cite{seybold2018thesis}) introduced the {\it scale-diversity} measure $\mathcal{D}$ for (2-dimensional) GMMN instances, and gave an $O(\mathcal{D})$-approximation algorithm.
Because $\mathcal{D} = O(\log n)$ is guaranteed, this also implies $O(\log n)$-approximation as with Das et al.~\cite{10.1007/s00453-017-0298-0}, which is the current best approximation ratio for the GMMN problem in general.

As another approach to the GMMN problem, Schnizler~\cite{michael2015masters} explored tractable cases by focusing on the intersection graphs of GMMN instances.
The intersection graph represents for which terminal pairs M-paths can intersect.
He showed that, when both the treewidth and the maximum degree of intersection graphs are bounded by constants, the GMMN problem can be solved in polynomial time by dynamic programming (see Table \ref{tab:previous_result}).
His algorithm heavily depends on the degree bound, and it is natural to ask whether we can remove it, e.g., whether the GMMN problem is difficult even if the intersection graph is restricted to a tree without any degree bound.

In this paper, we give an answer to this question.
Specifically, as the first tractable case without any degree bound in the intersection graphs,
we provide a polynomial-time algorithm for the star case by reducing it to the longest path problem in directed acyclic graphs.

\begin{theorem}\label{thm:star}
There exists an $O(n^2)$-time algorithm for the GMMN problem when the intersection graph is restricted to a star.
\end{theorem}

Then, we extend it to the general tree case based on a dynamic programming (DP) approach inspired by and improving Schnizler's algorithm~\cite{michael2015masters}.

\begin{theorem}\label{thm:tree}
There exists an $O(n^5)$-time algorithm for the GMMN problem when the intersection graph is restricted to a tree.
\end{theorem}

The above algorithm involves two types of DPs, which are nested.
We furthermore improve its running time by reducing the computational cost of inner DPs, and obtain the following result.

\begin{theorem}\label{thm:speedup}
There exists an $O(n^3)$-time algorithm for the GMMN problem when the intersection graph is restricted to a tree.
\end{theorem}

Furthermore, we show that the cycle case can be solved by solving the tree case $O(n)$ times.
This fact is shown as Proposition~\ref{prop:cycle} in a generalized form from cycles to triangle-free pseudotrees, where a triangle is a cycle consisting of three vertices and a pseudotree is a connected graph that contains at most one cycle.\footnote{Precisely, a triangle itself is not a triangle-free pseudotree, but its size is trivially bounded by a constant. In contrast, the size of a pseudotree containing a triangle is unbounded, and it remains open whether such a case is tractable or not. See Section~\ref{sec:specialization} for why the triangle-freeness is crucial in our approach.}
Combining this with Theorem~\ref{thm:speedup}, we obtain the following result.

\begin{corollary}\label{cor:cycle}
There exists an $O(n^4)$-time algorithm for the GMMN problem when the intersection graph is restricted to a cycle (or a triangle-free pseudotree).
\end{corollary}

We also improve the time complexity for the general case as in Table~\ref{tab:previous_result}.
The dependency on the maximum degree is substantially improved, but it is still exponential.
In addition, the approach is apart from the above main results and is also a straightforward improvement from Schnizler's result for the tree case.
For these reasons, we just sketch this result in the appendix. 

\begin{table}[t]
    \centering
    \caption{Exactly solvable cases classified by the class of intersection graphs, whose treewidth and maximum degree are denoted by $\tw$ and $\Delta$, respectively.}
    \begin{tabular}{cc}\hline
         Class & Time Complexity \\\hline\hline
         $\tw = O(1)$, $\Delta = O(1)$ & $O(n^{4\Delta(\Delta+1)(\tw+1)+2})$ \cite{michael2015masters}\\\hline
         Trees ($\tw = 1$, $\Delta = O(1))$ & $O(n^{4\Delta^2 + 1})$ \cite{michael2015masters} \\\hline
         Cycles ($\tw = \Delta = 2$) & $O(n^{25})$ \cite{michael2015masters} \\\hline\hline 
         $\tw = O(1)$, $\Delta = O(1)$ & $O(n^{2\Delta(\tw+1)+1})$ (Theorem~\ref{thm:tree-decomposition}) \\\hline
         Stars ($\tw = 1$, $\Delta = n - 1$) & $O(n^2)$ (Theorem~\ref{thm:star}) \\\hline
         Trees ($\tw = 1$) & $O(n^3)$ (Theorem~\ref{thm:speedup}) \\\hline
         Cycles ($\tw = \Delta = 2$) & $O(n^4)$ (Corollary~\ref{cor:cycle}) \\\hline
    \end{tabular}
    \label{tab:previous_result}    
\end{table}

\subsection*{Related work}
The MMN problem was first introduced by Gudmundsson, Levcopoulos, and Narashimhan~\cite{gudmundsson2001approximating}.
They gave 4- and 8-approximation algorithms running in $O(n^3)$ and $O(n \log n)$ time, respectively.
The current best approximation ratio is 2, which was obtained independently by Chepoi et al.~\cite{chepoi2008rounding} using an LP-ronding technique, by Nouioua~\cite{Nouioua2005EnveloppesDP} using a primal-dual scheme, and by Guo, Sun, and Zhu~\cite{guo2008greedy} using a greedy method.

The RSA problem is another important special case of the GMMN problem. 
In this problem, given a set of terminals in $\mathbb{R}^2$, we are required to find a minimum-length network that contains an M-path between the origin and every terminal. 
The RSA problem was first studied by Nastansky, Selkow, and Stewart~\cite{nastansky1974cost} in 1974.
The complexity of the RSA problem had been open for a long time, and Shi and Su~\cite{shi2005rectilinear} showed that the decision version is strongly NP-complete after three decades.
Rao, Sadayappan, Hwang, and Shor~\cite{rao1992rectilinear} proposed a 2-approximation algorithm that runs in $O(n \log n)$ time.
Lu and Ruan~\cite{lu2000polynomial} and Zachariasen~\cite{zachariasen2000approximation} independently obtained PTASes, which are both based on Arora's technique~\cite{arora2003approximation} of building a PTAS for the metric Steiner tree problem.

\subsection*{Organization}
The rest of this paper is organized as follows.
In Section~\ref{sec:preliminary}, we describe necessary definitions and notations.
In Section~\ref{sec:star}, we present an algorithm for the star case and prove Theorem~\ref{thm:star}.
In Section \ref{sec:tree}, based on a DP approach, we extend our algorithm to the tree case and prove Theorem~\ref{thm:tree}. 
Then, in Section~\ref{sec:speedup}, we improve the algorithm shown in Section~\ref{sec:tree} by reducing the computational cost of solving subproblems in our DP and prove Theorem~\ref{thm:speedup}.
Finally, in Section~\ref{sec:cycle}, we show that any cycle (or triangle-free pseudotree) instance can be reduced to $O(n)$ tree instances, which implies Corollary~\ref{cor:cycle}.
We also discuss an improvement on the general case and another observation in the appendix.
\section{Preliminaries}\label{sec:preliminary}
\subsection{Problem Formulation}
For a point $p \in \mathbb{R}^2$, we denote by $p_x$ and $p_y$ its $x$- and $y$-coordinates, respectively, i.e., $p = (p_x, p_y)$.
Let $p, q \in \mathbb{R}^2$ be two points.
We write $p \leq q$ if both $p_x \leq q_x$ and $p_y \leq q_y$ hold.
We define two points
\begin{align}
  p \wedge q &= \left(\min\left\{p_x, q_x\right\},\, \min\left\{p_y, q_y\right\}\right),\\
  p \vee q &= \left(\max\left\{p_x, q_x\right\},\, \max\left\{p_y, q_y\right\}\right).
\end{align}
We denote by $pq$ the segment whose endpoints are $p$ and $q$, and by $\|pq\|$ its length, i.e., $pq = \{ \alpha p + (1 - \alpha) q \mid \alpha \in [0, 1] \}$ and $\|pq\| = \sqrt{(p_x - q_x)^2 + (p_y - q_y)^2}$.
We also define $d_x(p, q) = |p_x - q_x|$ and $d_y(p, q) = |p_y - q_y|$,
and denote by $d(p, q)$ the \emph{Manhattan distance} between $p$ and $q$, i.e., $d(p, q) = d_x(p, q) + d_y(p, q)$.
Note that $\|pq\| = d(p, q)$ if and only if $p_x = q_x$ or $p_y = q_y$, and then the segment $pq$ is said to be \emph{vertical} or \emph{horizontal}, respectively, and \emph{axis-aligned} in either case.

A \emph{(geometric) network} $N$ in $\mathbb{R}^2$ is a finite simple graph with a vertex set $V(N) \subseteq \mathbb{R}^2$ and an edge set $E(N) \subseteq \binom{V(N)}{2} = \{ \{p, q\} \mid p, q \in V(N),~p \neq q \}$,
where we often identify each edge $\{p, q\}$ with the corresponding segment $pq$.
The \emph{length} of $N$ is defined as $\|N\| = \sum_{\{p, q\} \in E(N)} \|pq\|$.
For two points $s, t \in \mathbb{R}^2$, a \emph{path} $\pi$ between $s$ and $t$ (or an \emph{$s$--$t$ path}) is a network of the following form: 
\begin{align}
    V(\pi) &= \{s = p_0, p_1, p_2, \dots, p_k = t\},\\
    E(\pi) &= \bigl\{\{p_{i-1}, p_i\} \mid i \in [k] \bigr\},
\end{align}
where $[k] = \{1, 2, \dots, k\}$ for a nonnegative integer $k$.
An $s$--$t$ path $\pi$ is called a \emph{Manhattan path} (or an \emph{M-path}) for a pair $(s, t)$ if
every edge $\{p_{i-1}, p_i\} \in E(\pi)$ is axis-aligned and $\|\pi\| = d(s, t)$ holds.

We are now ready to state our problem formally.

\begin{problem}[Generalized Minimum Manhattan Network (GMMN)]
\begin{description}
\setlength{\itemsep}{0mm}
\item[]
\item[Input:] A set $\GMMN$ of $n$ pairs of points in $\mathbb{R}^2$.
\item[Goal:] Find a minimum-length network $N$ in $\mathbb{R}^2$ that consists of axis-aligned edges and contains a Manhattan path for every pair $(s, t) \in \GMMN$.
\end{description}
\end{problem}

Throughout this paper, when we write a pair $(p, q) \in \mathbb{R}^2 \times \mathbb{R}^2$, we assume $p_x \leq q_x$ (by swapping if necessary).
A pair $(p, q)$ is said to be \emph{regular} if $p_y \leq q_y$, and \emph{flipped} if $p_y \geq q_y$.
In addition, if $p_x = q_x$ or $p_y = q_y$, then there exists a unique M-path for $(p, q)$ and we call such a pair \emph{degenerate}.

\subsection{Restricting a Feasible Region to the Hanan Grid}
For a GMMN instance $\GMMN$, we denote by $\mathcal{H}(\GMMN)$ the \emph{Hanan grid}, which is a grid network in $\mathbb{R}^2$ consisting of vertical and horizontal lines through each point appearing in $\GMMN$.
More formally, it is defined as follows (see Figure~\ref{fig:GMMN_Hanangrid}):
\begin{align}
    V(\mathcal{H}(\GMMN)) &= \left(\bigcup_{(s, t) \in \GMMN}\{ s_x, t_x \}\right) \times \left(\bigcup_{(s, t) \in \GMMN}\{ s_y, t_y \}\right) \subseteq \mathbb{R}^2,\\
    E(\mathcal{H}(\GMMN)) &= \bigl\{\{p, q\} \in \textstyle\binom{V(\mathcal{H}(\GMMN))}{2}\mid \|pq\| = d(p,q),~pq \cap V(\mathcal{H}(\GMMN)) = \{p, q\} \bigr\}.
\end{align}
Note that $\mathcal{H}(\GMMN)$ is an at most $2n \times 2n$ grid network. 
It is not difficult to see that, for any GMMN instance $\GMMN$, at least one optimal solution is contained in the Hanan grid $\mathcal{H}(\GMMN)$ as its subgraph (cf.~\cite{funke2014generalized}).

\begin{figure}
    \centering
    \includegraphics[width=0.45\hsize]{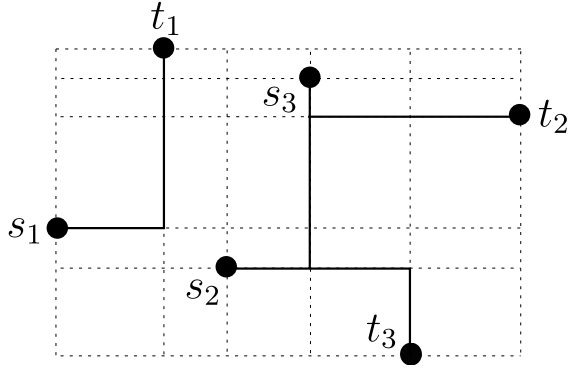}
    \caption{An optimal solution (solid) to a GMMN instance $\{(s_1, t_1),(s_2, t_2),(s_3, t_3)\}$ lies on the Hanan grid (dashed), where $(s_1, t_1)$ and $(s_2, t_2)$ are regular pairs and $(s_3, t_3)$ is a flipped pair.}
    \label{fig:GMMN_Hanangrid}
\end{figure}

For each pair $v = (p, q) \in V(\mathcal{H}(\GMMN)) \times V(\mathcal{H}(\GMMN))$, we denote by $\Pi_\GMMN(v)$ or $\Pi_\GMMN(p, q)$ the set of all M-paths for $v$ that are subgraphs of the Hanan grid $\cH(\GMMN)$.
By the problem definition, we associate each $n$-tuple of M-paths, consisting of an M-path $\pi_v \in \Pi_\GMMN(v)$ for each $v \in \GMMN$, with a feasible solution $N = \bigcup_{v \in \GMMN} \pi_v$ on $\cH(\GMMN)$, where the union of networks is defined by the set unions of the vertex sets and of the edge sets.
Moreover, each minimal feasible (as well as optimal) solution on $\cH(\GMMN)$ must be represented in this way.
Based on this correspondence, we abuse the notation as $N = (\pi_v)_{v \in \GMMN} \in \prod_{v \in \GMMN}\Pi_\GMMN(v)$, and define $\mathrm{Feas}(\GMMN)$ and $\mathrm{Opt}(\GMMN)$ as the sets of feasible solutions covering all minimal ones and of all optimal solutions, respectively, on $\cH(\GMMN)$, i.e.,
\begin{align}
    \Feas(\GMMN) &= \prod_{v \in \GMMN}\Pi_\GMMN(v),\\
    \Opt(\GMMN) &= \argmin\{ \|N\| \mid N \in \Feas(\GMMN) \}.
\end{align}

Thus, we have restricted a feasible region of a GMMN instance $\GMMN$ to the Hanan grid $\cH(\GMMN)$.
In other words, the GMMN problem reduces to finding a network $N = (\pi_v)_{v \in \GMMN} \in \Opt(\GMMN)$ as an $n$-tuple of M-paths in $\Feas(\GMMN)$.

\subsection{Specialization Based on Intersection Graphs}\label{sec:specialization}
The \emph{bounding box} of a pair $v = (p, q) \in \mathbb{R}^2 \times \mathbb{R}^2$ indicates the rectangle region
\begin{align}
    \{z \in \mathbb R^2 \mid p \wedge q \le z \le p \vee q \},
\end{align}
and we denote it by $B(v)$ or $B(p, q)$.
Note that $B(p, q)$ is the region where an M-path for $(p, q)$ can exist.
For a GMMN instance $\GMMN$ and a pair $v \in \GMMN$, we denote by $\mathcal{H}(\GMMN, v)$ the subgraph of the Hanan grid $\mathcal{H}(\GMMN)$ induced by $V(\mathcal{H}(\GMMN)) \cap B(v)$.
We define the \emph{intersection graph} $\mathrm{IG}[\GMMN]$ of $\GMMN$ by
\begin{align}
    V(\mathrm{IG}[\GMMN]) &= \GMMN, \\
    E(\mathrm{IG}[\GMMN]) &= \bigl\{\{u, v\}\in \textstyle\binom{\GMMN}{2} \mid E(\mathcal{H}(\GMMN, u)) \cap E(\mathcal{H}(\GMMN, v)) \ne \emptyset \bigr\}.
\end{align}

The intersection graph $\IG[\GMMN]$ intuitively represents how a GMMN instance $\GMMN$ is complicated in the sense that, for each $u, v \in \GMMN$,
an edge $\{u, v\} \in E(\mathrm{IG}[\GMMN])$ exists if and only if two M-paths $\pi_u \in \Pi_\GMMN(u)$ and $\pi_v \in \Pi_\GMMN(v)$ can share some segments, which saves the total length of a network in $\Feas(\GMMN)$.\footnote{We remark that our definition of the intersection graph is slightly different from Schnizler's one~\cite{michael2015masters}, which regards two pairs as adjacent even when their bounding boxes share exactly one point. We employ our definition because M-paths for such pairs cannot share any nontrivial segment. This difference itself expands tractable situations, and sometimes requires more careful arguments due to shared points of nonadjacent pairs (in particular, the corners of their bounding boxes).}
In particular, if $\IG[\GMMN]$ contains no triangle, then no segment can be shared by M-paths for three different pairs in $\GMMN$, and hence $N \in \Feas(\GMMN)$ is optimal (i.e., $\|N\|$ is minimized) if and only if the total length of segments shared by two M-paths in $N$ is maximized.

We denote by GMMN[$\cdots$] the GMMN problem with restriction on the intersection graph of the input; e.g., $\IG[P]$ is restricted to a tree in GMMN[Tree].
Each restricted problem is formally stated in the relevant section.
\section{An $O(n^2)$-Time Algorithm for GMMN[Star]}\label{sec:star}
In this section, as a step to GMMN[Tree], we present an $O(n^2)$-time algorithm for GMMN[Star], which is formally state as follows.

\begin{problem}[{GMMN[Star]}]
\begin{description}
\setlength{\itemsep}{0mm}
\item[]
\item[Input:] A set $\GMMN \subseteq \mathbb{R}^2 \times \mathbb{R}^2$ of $n$ pairs whose intersection graph $\IG[\GMMN]$ is a star, whose center is denoted by $r = (s, t) \in \GMMN$.
\item[Goal:] Find an optimal network $N = (\pi_v)_{v \in \GMMN} \in \Opt(\GMMN)$.
\end{description}
\end{problem}

A crucial observation for GMMN[Star] is that an M-path $\pi_l \in \Pi_\GMMN(l)$ for each leaf pair $l \in \GMMN - r$ can share some segments only with an M-path $\pi_r \in \Pi_\GMMN(r)$ for the center pair $r$.
Hence, minimizing the length of $N = (\pi_v)_{v \in \GMMN}\in \Feas(\GMMN)$ is equivalent to maximizing the total length of segments shared by two M-paths $\pi_r$ and $\pi_l$ for $l \in \GMMN - r$.

In Section~\ref{sec:sharable}, we observe that, for each leaf pair $l \in \GMMN - r$,
once we fix where an M-path $\pi_r \in \Pi_\GMMN(r)$ for $r$ enters and leaves the bounding box $B(l)$,
the maximum length of segments that can be shared by $\pi_r$ and $\pi_l \in \Pi_\GMMN(l)$ is easily determined.
Thus, GMMN[Star] reduces to finding an optimal M-path $\pi_r \in \Pi_\GMMN(r)$ for the center pair $r = (s, t)$, and in Section~\ref{sec:reduction}, we formulate this task as the computation of a longest $s$--$t$ path in an auxiliary directed acyclic graph (DAG), which is constructed from the subgrid $\cH(\GMMN, r)$.
As a result, we obtain an exact algorithm that runs in linear time in the size of auxiliary graphs, which are simplified so that it is always $O(n^2)$ in Section~\ref{sec:star:runtime_analysis}.

\subsection{Observation on Sharable Segments}\label{sec:sharable}
Without loss of generality, we assume that the center pair $r = (s, t)$ is regular, i.e., $s \leq t$.
Fix an M-path $\pi_r \in \Pi_\GMMN(r)$ and a leaf pair $l = (s_l, t_l) \in \GMMN - r$.
Obviously, if $\pi_r$ is disjoint from the bounding box $B(l)$, then any M-path $\pi_l \in \Pi_\GMMN(l)$ cannot share any segment with $\pi_r$.
Suppose that $\pi_r$ intersects $B(l)$, and let $\pi_r[l]$ denote the intersection $\pi_r \cap \cH(\GMMN, l)$.
Let $v = (p, q)$ be the pair of two vertices on $\pi_r$ such that $\pi_r[l]$ is a $p$--$q$ path, and we call $v$ the \emph{in-out pair} of $\pi_r$ for $l$.
As $\pi_r \in \Pi_\GMMN(r)$, we have $\pi_r[l] \in \Pi_\GMMN(v)$, and $v$ is also regular (recall the assumption $p_x \leq q_x$).
Moreover, for any M-path $\pi_v \in \Pi_\GMMN(v)$, the network $\pi'_r$ obtained from $\pi_r$ by replacing its subpath $\pi_r[l]$ with $\pi_v$ is also an M-path for $r$ in $\Pi_\GMMN(r)$.
Since $B(v) \subseteq B(l)$ does not intersect $B(l')$ for any other leaf pair $l' \in \GMMN \setminus \{r, l\}$, once $v = (p, q)$ is fixed, we can freely choose an M-path $\pi_v \in \Pi_\GMMN(v)$ instead of $\pi_r[l]$ for maximizing the length of segments shared with some $\pi_l \in \Pi_\GMMN(l)$.
For each possible in-out pair $v = (p, q)$ of M-paths in $\Pi_\GMMN(r)$
(the sets of those vertices $p$ and $q$ are formally defined in Section~\ref{sec:reduction} as $V_{\llcorner}(r, l)$ and $V_{\urcorner}(r, l)$, respectively),
we denote by $\gamma(l, p, q)$ the maximum length of segments shared by two M-paths for $l$ and $v = (p, q)$, i.e.,
\begin{align}
\gamma(l, p, q) = \max\left\{\|\pi_l \cap \pi_v\| \mid \pi_l \in \Pi_\GMMN(l),~\pi_v \in \Pi_\GMMN(p, q)\right\}. \label{eq:sharable_length}
\end{align}
Then, the following lemma is easily observed (see Figure~\ref{fig:share}).

\begin{lemma}\label{lem:sharable_segments}
For every leaf pair $l \in \GMMN - r$, the following properties hold.
\begin{enumerate}[label={(\arabic*)},labelindent=\parindent,leftmargin=*]
    \item[$(1)$] If $l$ is a regular pair, $\gamma(l, p, q) = d(p, q)\ \bigl(= d_x(p, q) + d_y(p, q)\bigr)$.
    \item[$(2)$] If $l$ is a flipped pair, $\gamma(l, p, q) = \max\left\{d_x(p, q),\, d_y(p, q)\right\}$. 
\end{enumerate}
\end{lemma}

\begin{figure}[tb]
\begin{tabular}{cc}
\begin{minipage}{0.45\hsize}
    \centering
    \includegraphics[width=0.8\hsize]{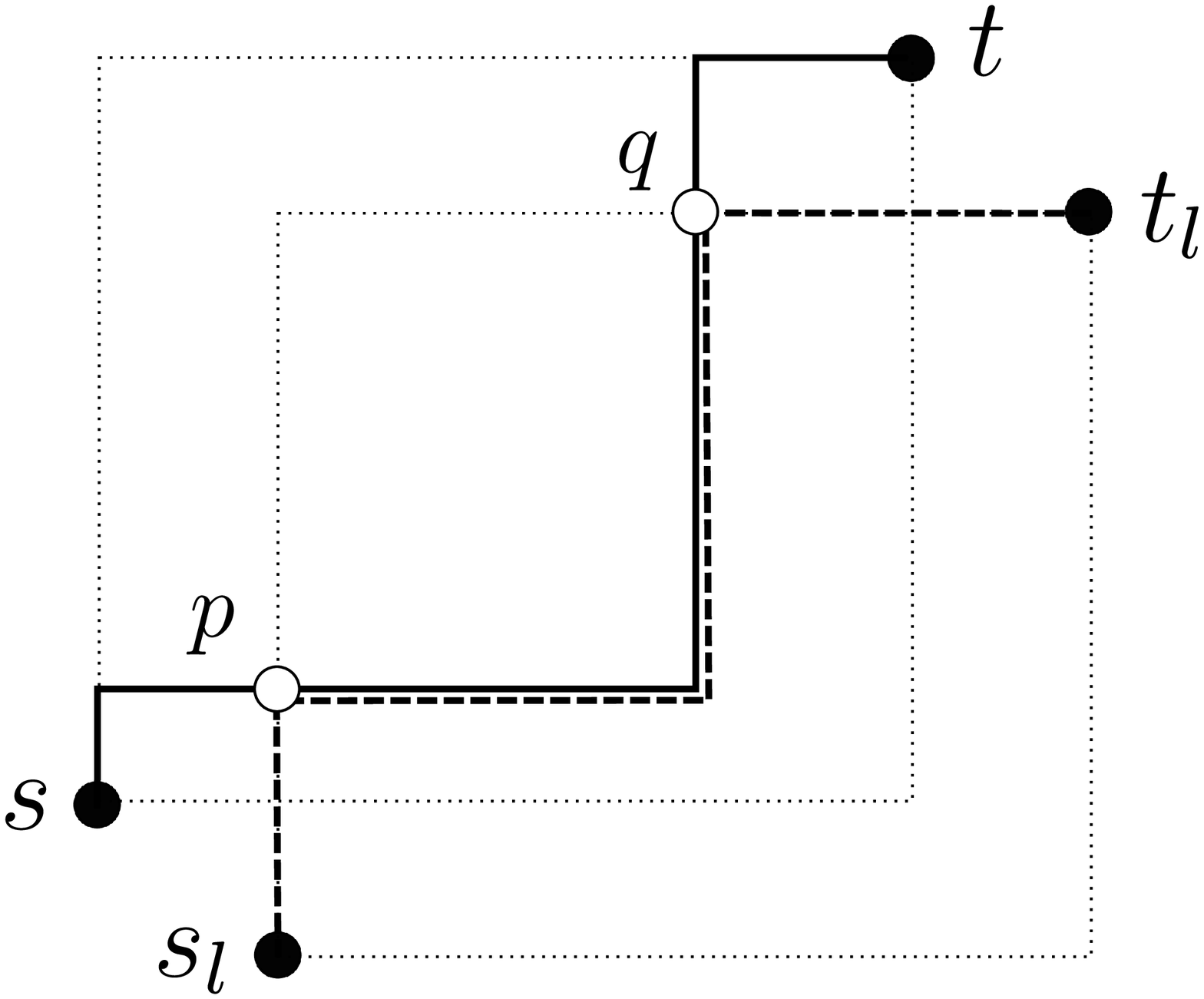}
    \subcaption{}
    \label{fig:share1}
\end{minipage}&
\begin{minipage}{0.45\hsize}
    \centering
    \includegraphics[width=0.8\hsize]{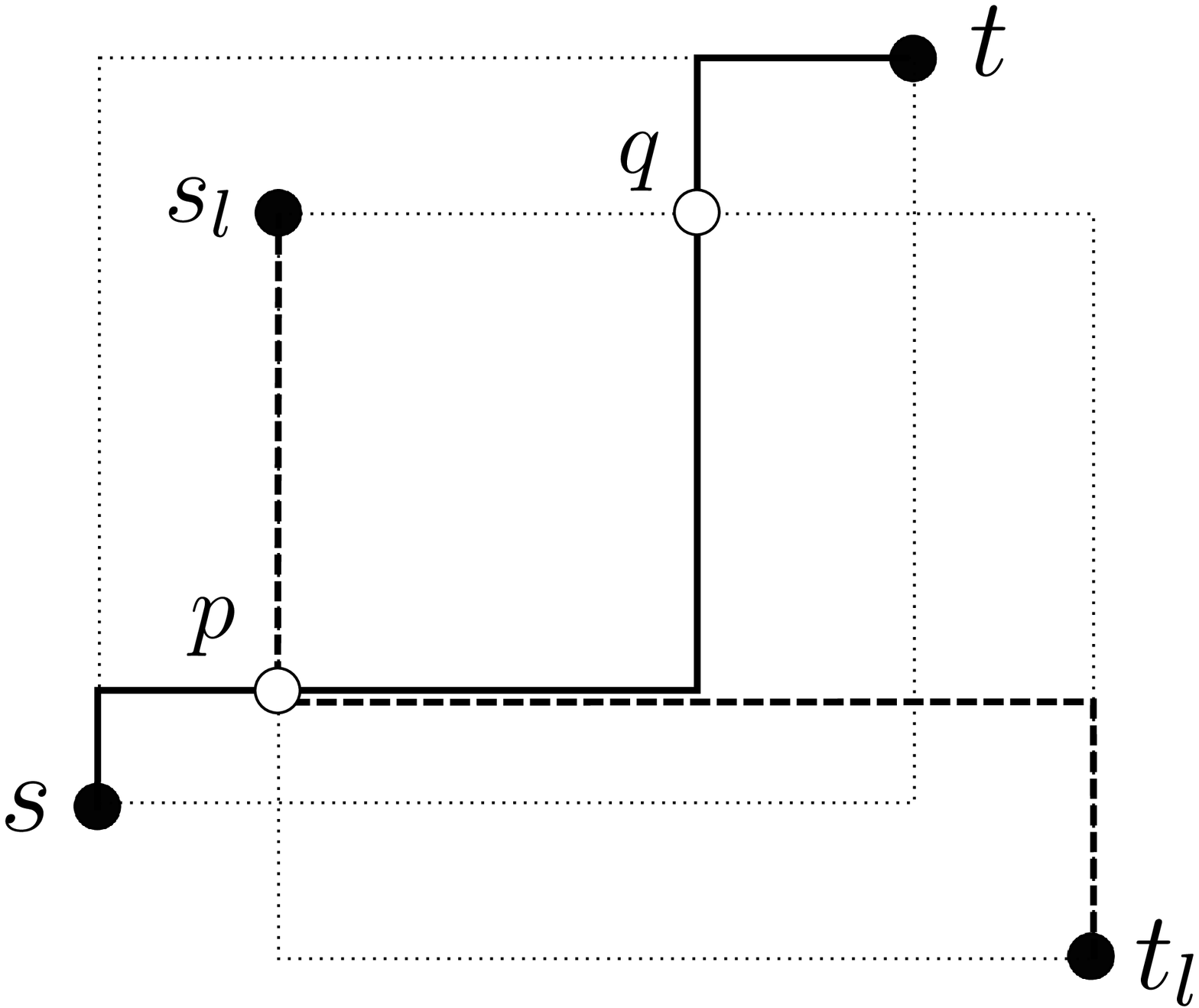}
    \subcaption{}
    \label{fig:share2}
\end{minipage}
\end{tabular}
\caption{(a) If $l = (s_l, t_l)$ is a regular pair, for any $\pi_v \in \Pi_\GMMN(p, q)$, some $\pi_l \in \Pi_\GMMN(l)$ completely includes $\pi_v$. (b) If $l= (s_l, t_l)$ is a flipped pair, while any $\pi_l \in \Pi_\GMMN(l)$ cannot contain both horizontal and vertical segments of any $\pi_v \in \Pi_\GMMN(p, q)$, one can choose $\pi_v \in \Pi_\GMMN(p, q)$ so that the whole of either horizontal or vertical segments of $\pi_v$ can be included in some $\pi_l \in \Pi_\GMMN(l)$.}\label{fig:share}
\end{figure}

\subsection{Reduction to the Longest Path Problem in DAGs}\label{sec:reduction}
In this section, we reduce GMMN[Star] to the longest path problem in DAGs. 
Let $\GMMN$ be a GMMN[Star] instance and $r = (s, t) \in \GMMN$ ($s \le t$) be the center of $\IG[\GMMN]$,
and we construct an auxiliary DAG $G$ from the subgrid $\mathcal{H}(\GMMN, r)$ as follows (see Figure~\ref{fig:dag_overview}).

\begin{figure}[tb]
\begin{tabular}{cc}
\begin{minipage}{0.45\hsize}
    \centering
    \includegraphics[width=0.8\hsize]{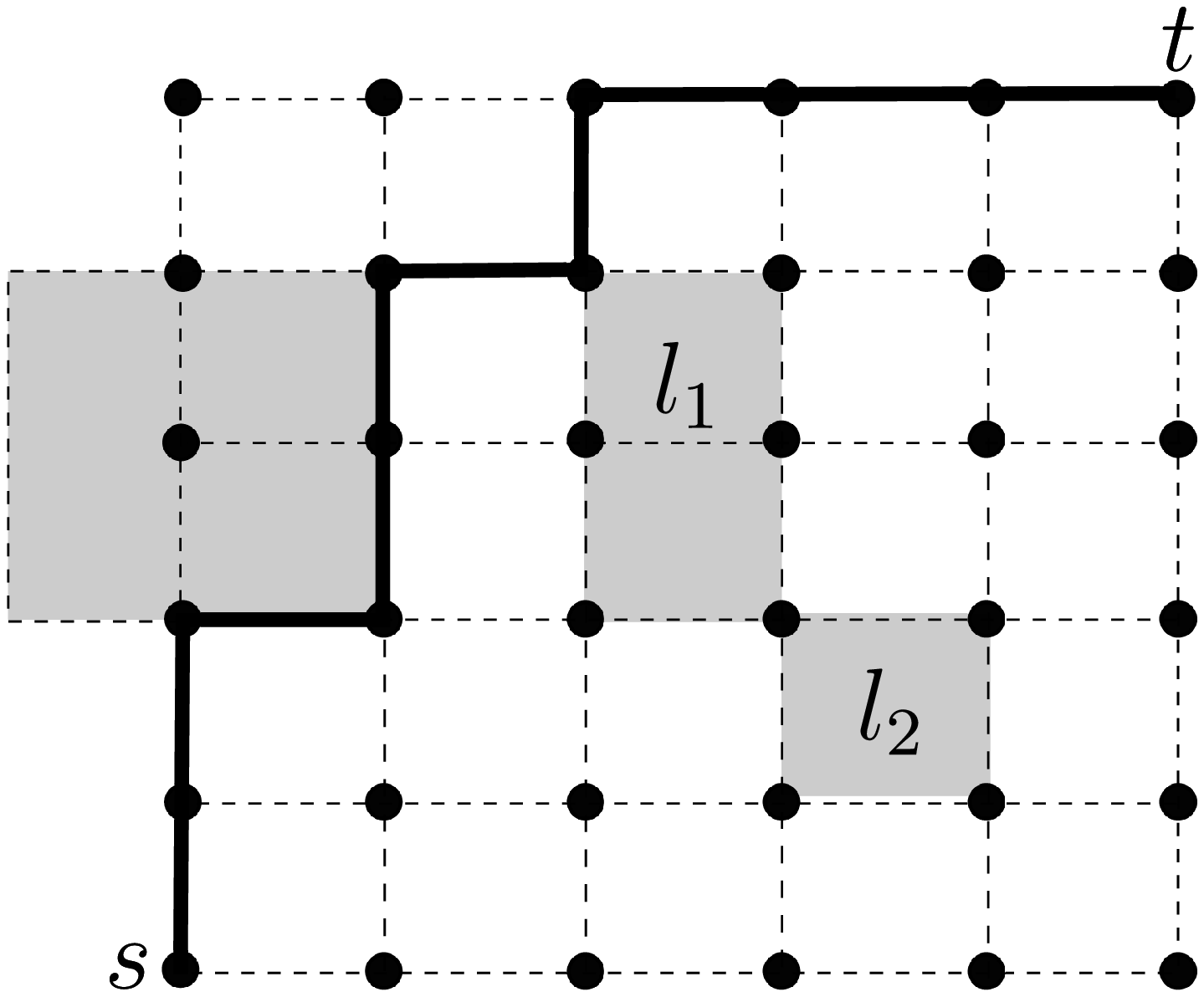}
    \subcaption{}
\end{minipage}& 
\begin{minipage}{0.45\hsize}
    \centering
    \includegraphics[width=0.8\hsize]{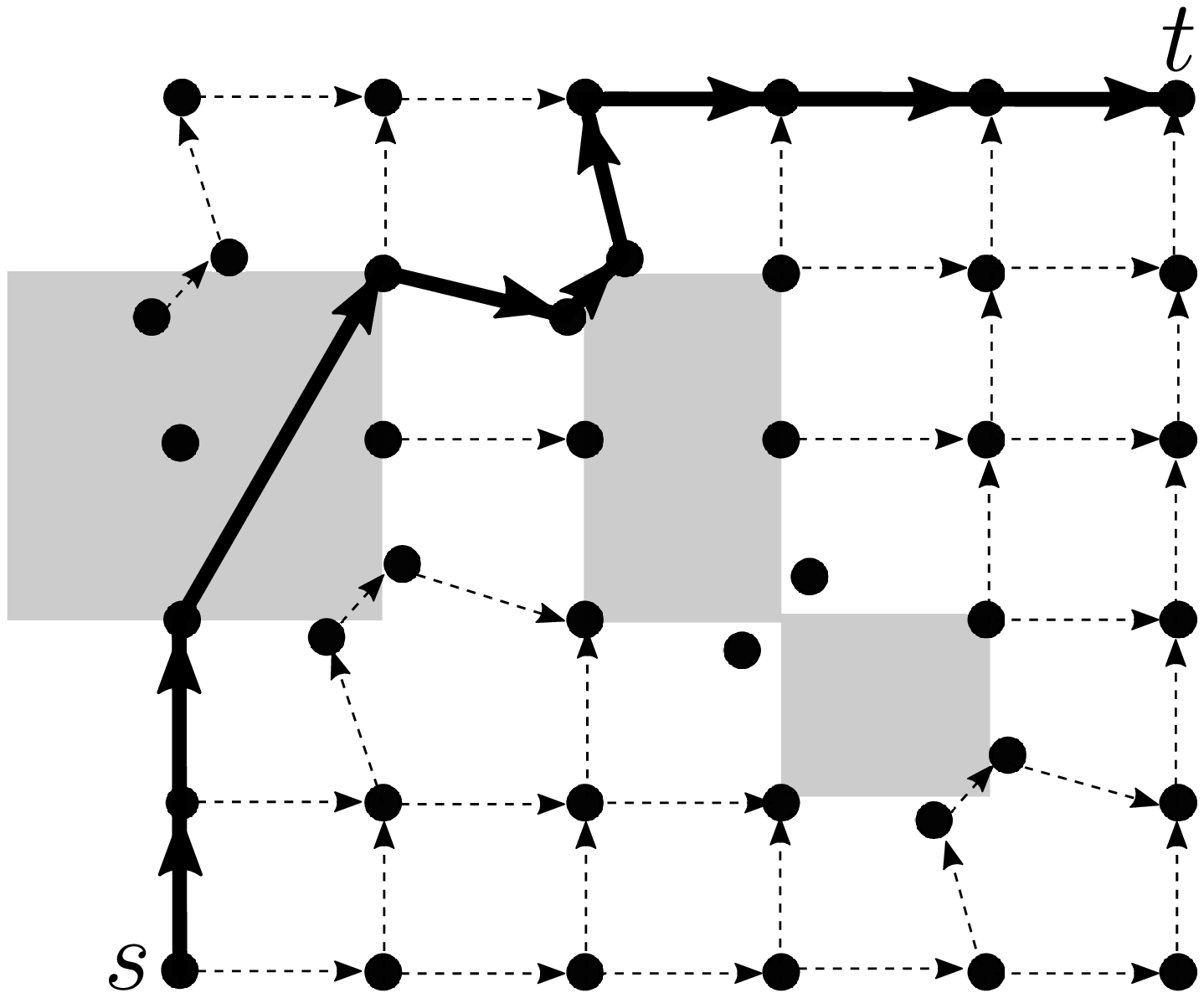}
    \subcaption{}
\end{minipage}\\
\begin{minipage}{0.4\hsize}
    \includegraphics[width=\hsize]{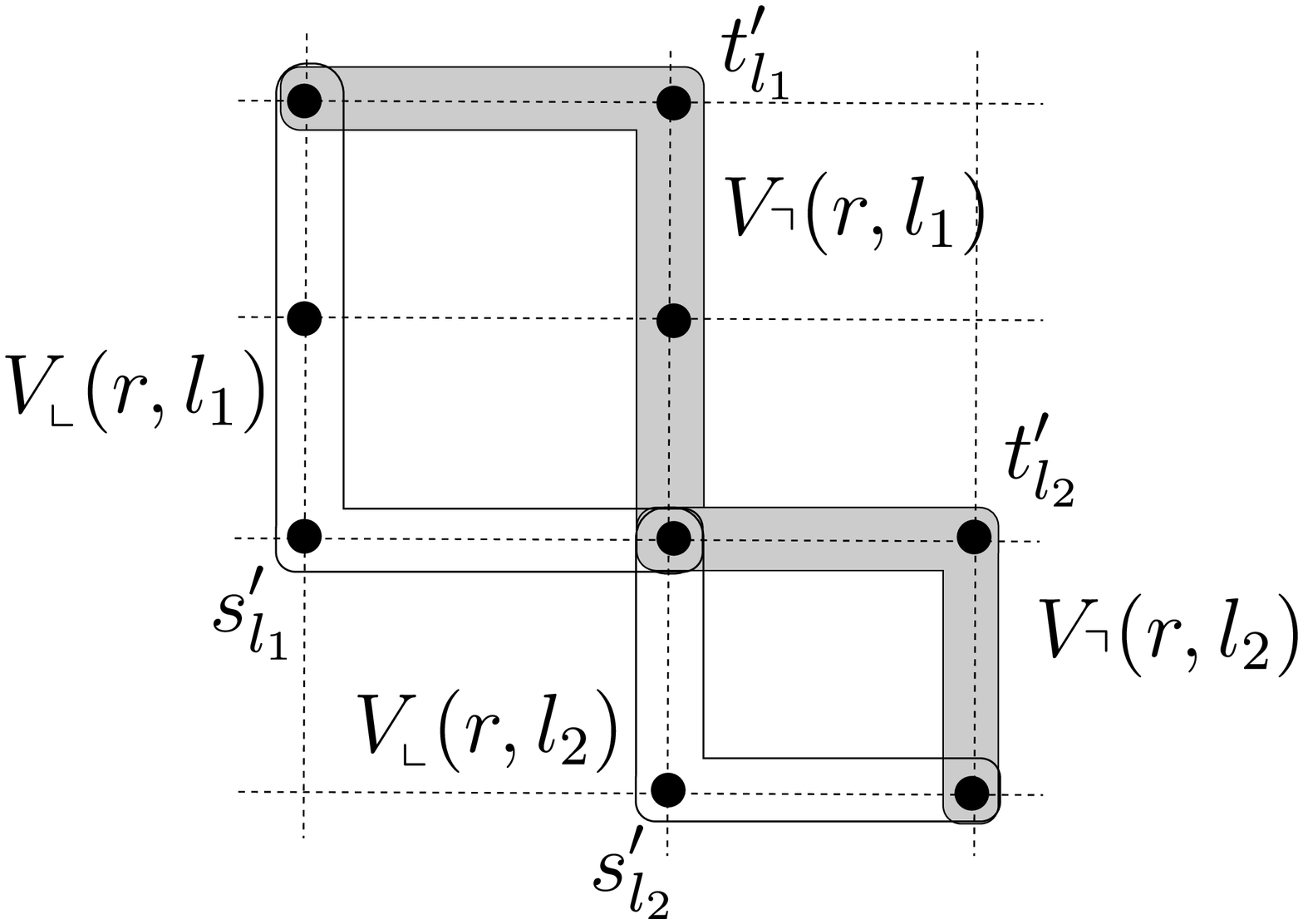}
    \subcaption{}
\end{minipage}&
\begin{minipage}{0.5\hsize}
    \centering
    \includegraphics[width=\hsize]{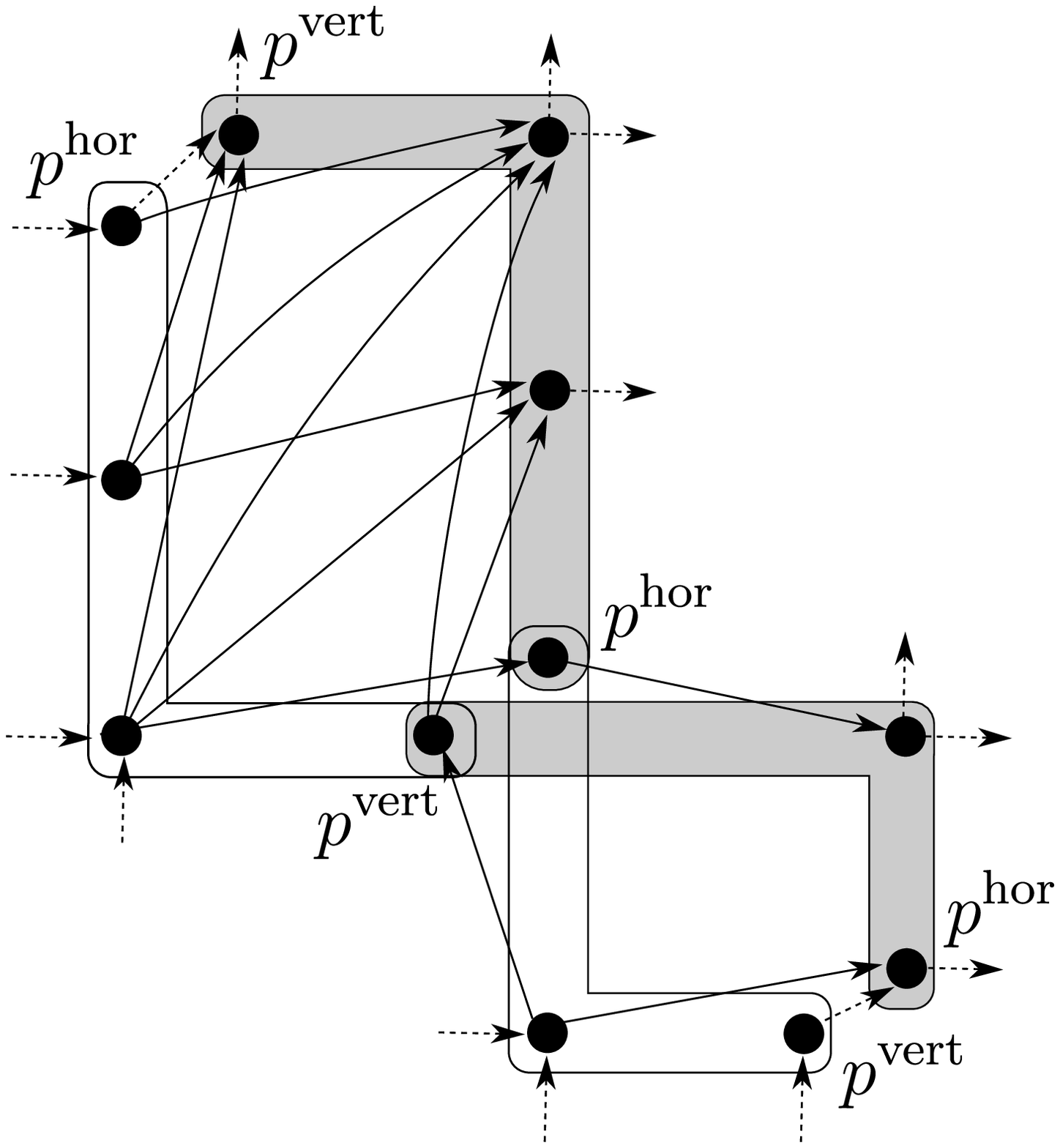}
    \subcaption{}
\end{minipage}
\end{tabular}
\caption{(a) An M-path for $r$ in the subgrid $\cH(\GMMN, r)$.
(b) The corresponding directed $s$--$t$ path in the auxiliary DAG $G$, where the dashed arcs are of length $0$.
(c) The boundary vertex sets for leaf pairs.
(d) The corresponding parts in $G$, where the length of each interior arc $(p, q)$ is $\gamma(l, p, q)$ for $p \in V_\llcorner(r, l)$ and $q \in V_\urcorner(r, l)$. }
\label{fig:dag_overview}
\end{figure}

First, for each edge $e = \{p, q\} \in E(\mathcal{H}(\GMMN, r))$ with $p \leq q$ (and $p \neq q$), we replace $e$ with an arc $(p, q)$ of length $0$.
For each leaf pair $l \in \GMMN - r$, let $s'_l$ and $t'_l$ denote the lower-left and upper-right corners of $B(r) \cap B(l)$, respectively, so that $(s'_l, t'_l)$ is a regular pair with $B(s'_l, t'_l) = B(r) \cap B(l)$.
If $(s'_l, t'_l)$ is degenerate, then we change the length of each arc $(p, q)$ with $p, q \in V(\mathcal{H}(\GMMN, r) \cap B(l))$ from $0$ to $\|pq\|$, which clearly reflects the (maximum) sharable length in $B(l)$.
Otherwise, the bounding box $B(s'_l, t'_l) \subseteq B(l)$ has a nonempty interior, and we define four subsets of $V(\mathcal{H}(\GMMN, r) \cap B(l))$ as follows:
\begin{align}
 V_{\llcorner}(r, l) &= \{p \in V(\mathcal{H}(\GMMN, r) \cap B(l)) \mid p_x = (s'_l)_x~\text{or}~p_y = (s'_l)_y\}, \\
 V_{\urcorner}(r, l) &= \{q \in V(\mathcal{H}(\GMMN, r) \cap B(l)) \mid q_x = (t'_l)_x~\text{or}~q_y = (t'_l)_y\},\\
 \Vdd(r, l) &= V_{\llcorner}(r, l) \cap V_{\urcorner}(r, l),\\
 V_{\blacksquare}(r, l) &= V(\mathcal{H}(\GMMN, r) \cap B(l)) \setminus (V_{\llcorner}(r, l) \cup V_{\urcorner}(r, l))\\
 &= \{z \in V(\mathcal{H}(\GMMN, r) \cap B(l)) \mid (s'_l)_x < z_x < (t'_l)_x~\text{and}~(s'_l)_y < z_y < (t'_l)_y\}.
\end{align}

As $r$ is regular, any M-path $\pi_r \in \Pi_\GMMN(r)$ intersecting $B(l)$ enters it at some $p \in V_{\llcorner}(r, l)$ and leaves it at some $q \in V_{\urcorner}(r, l)$, and the maximum sharable length $\gamma(l, p, q)$ in $B(l)$ is determined by Lemma~\ref{lem:sharable_segments}.
We remove all the interior vertices in $V_{\blacksquare}(r, l)$ (with all the incident arcs) and all the boundary arcs $(p, q)$ with $p, q \in V_{\llcorner}(r, l) \cup V_{\urcorner}(r, l)$.
Instead, for each pair $(p, q)$ of $p \in V_{\llcorner}(r, l)$ and $q \in V_{\urcorner}(r, l)$ with $p \leq q$ and $p \neq q$, we add an interior arc $(p, q)$ of length $\gamma(l, p, q)$.
Let $E_\mathrm{int}(l)$ denote the set of such interior arcs for each nondegenerate pair $l \in \GMMN - r$.

Finally, we care about the corner vertices in $\Vdd(r, l)$, which can be used for cheating if $l$ is flipped as follows.
Suppose that $p \in \Vdd(r, l)$ is the upper-left corner of $B(l)$, and consider the situation when the in-out pair $(p', q')$ of $\pi_r \in \Pi_\GMMN(r)$ for $l$ satisfies $p'_x = p_x < q'_x$ and $p'_y < p_y = q'_y$.
Then, $(p', q')$ is not degenerate, and by Lemma~\ref{lem:sharable_segments}, the maximum sharable length in $B(l)$ is $\gamma(l, p', q') = \max\left\{d_x(p', q'),\, d_y(p', q')\right\}$ as it is represented by an interior arc $(p', q')$, but one can take another directed $p'$--$q'$ path that consists of two arcs $(p', p)$ and $(p, q')$ in the current graph, whose length is $d_y(p', p) + d_x(p, q') = d_y(p', q') + d_x(p', q') >\gamma(l, p', q')$.
To avoid such cheating, for each $p \in \Vdd(r, l)$, we divide it into two distinct copies $p^\mathrm{hor}$ and $p^\mathrm{vert}$ (which are often identified with its original $p$ unless we need to distinguish them), and replace the endpoint $p$ of each incident arc $e$ with $p^\mathrm{hor}$ if $e$ is horizontal and with $p^\mathrm{vert}$ if vertical (see Figure~\ref{fig:dag_overview} (d)).
In addition, when $p$ is not shared by any other leaf pair,\footnote{Note that $p$ can be shared as corners of two different leaf pairs due to our definition of the intersection graph, and then leaving one bounding box means entering the other straightforwardly.}
we add an arc $(p^\mathrm{hor}, p^\mathrm{vert})$ of length $0$ if $p$ is the upper-left corner of $B(s'_l, t'_l)$ and an arc $(p^\mathrm{vert}, p^\mathrm{hor})$ of length $0$ if the lower-right, which represents the situation when $\pi_r \in \Pi_\GMMN(r)$ intersects $B(l)$ only at $p$.

Let $G$ be the constructed directed graph, and denote by $\ell(e)$ the length of each arc $e \in E(G)$.
The following two lemmas complete our reduction (see Figure~\ref{fig:dag_overview} again).

\begin{lemma}\label{lem:star_dag}
The directed graph $G$ is acyclic.
\end{lemma}

\begin{proof} 
Almost all arcs are of form $(p, q)$ with $p \leq q$ and $p \neq q$.
The only exception is of form $(p^\mathrm{vert}, p^\mathrm{hor})$ or $(p^\mathrm{hor}, p^\mathrm{vert})$ for some $p \in \Vdd(r, l)$ with some $l \in \GMMN - r$, and at most one direction exists for each $p$ by definition.
Thus, $G$ contains no directed cycle.
\end{proof}

\begin{lemma}\label{lem:star_the_longest}
Any longest $s$--$t$ path $\pi^*_G$ in $G$ with respect to $\ell$ satisfies
\begin{align}
    \sum_{e \in E(\pi^*_G)} \ell(e) = \max_{\pi_r \in \Pi_\GMMN(r)}\left(\sum_{l \in \GMMN-r} \max_{\pi_l \in \Pi_\GMMN(l)}\|\pi_l \cap \pi_r\| \right).
\end{align}
\end{lemma}

\begin{proof}
Fix a directed $s$--$t$ path $\pi_G$ in $G$.
By the definition of $G$ and Lemma \ref{lem:star_dag}, for each nondegenerate pair $l \in \GMMN-r$, the path $\pi_G$ uses at most one interior arc in $E_\mathrm{int}(l)$, and any other arc has a trivially corresponding edge in $\mathcal{H}(\GMMN, r)$ (including edges in a degenerate pair).
For each $l$ with $E(\pi_G) \cap E_\mathrm{int}(l) \neq \emptyset$, let $e_l = (p, q)$ be the unique arc in $E(\pi_G) \cap E_\mathrm{int}(l)$.
By the definitions of $\ell$ and $\gamma$, we have
\begin{align}\label{eq:pi_e_l}
\ell(e_l) = \gamma(l, p, q) = \max\left\{\|\pi_l \cap \pi_{e_l}\| \mid \pi_l \in \Pi_\GMMN(l),~\pi_{e_l} \in \Pi_\GMMN(p, q)\right\},
\end{align}
and hence one can construct an M-path $\pi_r \in \Pi_\GMMN(r)$ by replacing each $e_l$ with some M-path $\pi_{e_l} \in \Pi_\GMMN(p, q)$ attaining \eqref{eq:pi_e_l} such that
\begin{align}\label{eq:pathlength}
\sum_{e \in E(\pi_G)} \ell(e) = \sum_{l \in \GMMN-r} \max_{\pi_l \in \Pi_\GMMN(l)}\|\pi_l \cap \pi_r\|.
\end{align}

To the contrary, for any M-path $\pi_r \in \Pi_\GMMN(r)$, by the definitions of $\gamma$ and $\ell$, one can construct a directed $s$--$t$ path $\pi_G$ in $G$ of length at least the right-hand side of \eqref{eq:pathlength}, and we are done.
\end{proof}

\subsection{Computational Time Analysis with Simplified DAGs}\label{sec:star:runtime_analysis}
A longest path in a DAG $G$ is computed in $O(|V(G)| + |E(G)|)$ time by dynamic programming.
Although the subgrid $\cH(\GMMN, r)$ has $O(n^2)$ vertices and edges, the auxiliary DAG $G$ constructed in Section~\ref{sec:reduction} may have much more arcs due to $E_\mathrm{int}(l)$, whose size is $\Theta(|V_{\llcorner}(r, l)| \cdot |V_{\urcorner}(r, l)|)$ and can be $\Omega(n^2)$. 
This, however, can be always reduced to linear by modifying the boundary vertices and the incident arcs appropriately in order to avoid creating diagonal arcs in $B(l)$.
In this section, we simplify $G$ to $G'$ with $O(n^2)$ vertices and edges, which completes the proof of Theorem~\ref{thm:star}.

\begin{figure}[tb]
\begin{tabular}{cc}
\begin{minipage}{0.4\hsize}
    \centering
    \includegraphics[width=0.6\hsize]{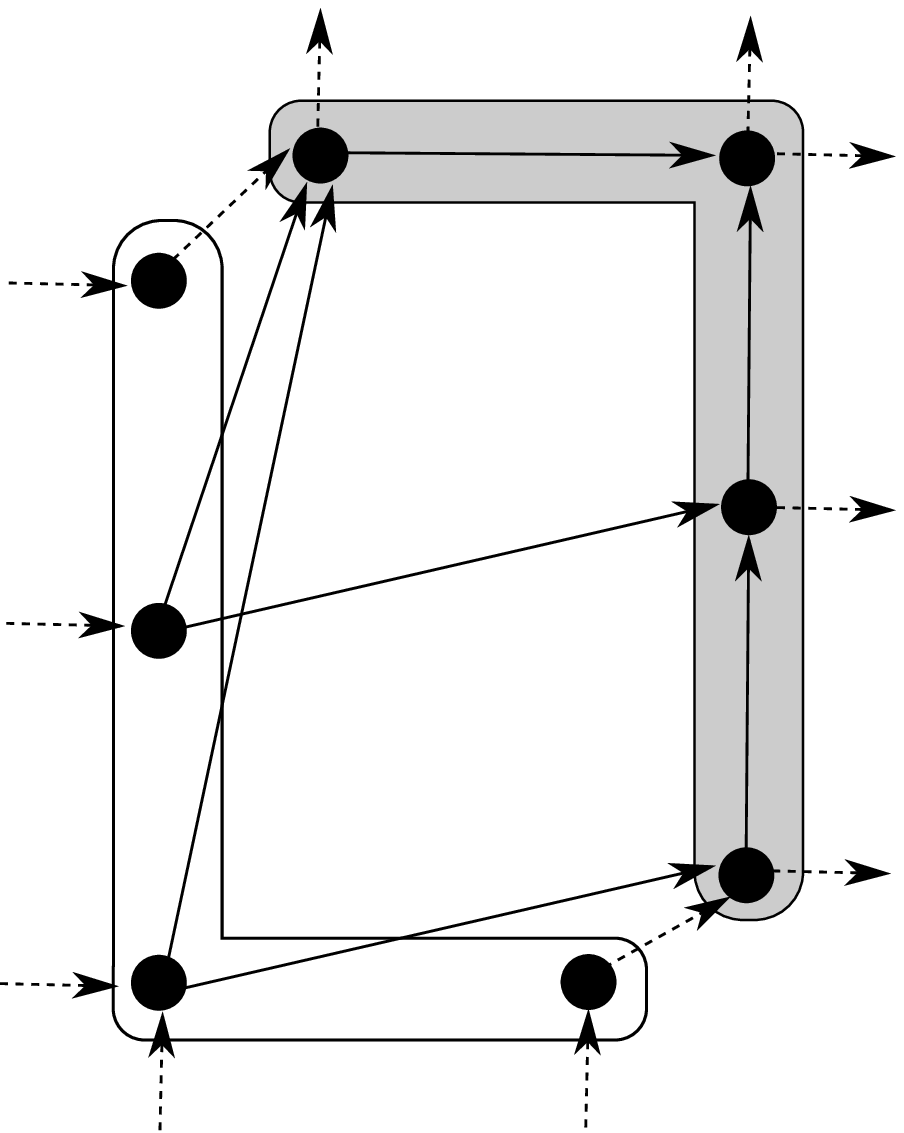}
    \subcaption{}
\end{minipage}&
\begin{minipage}{0.5\hsize}
    \centering
    \includegraphics[width=0.6\hsize]{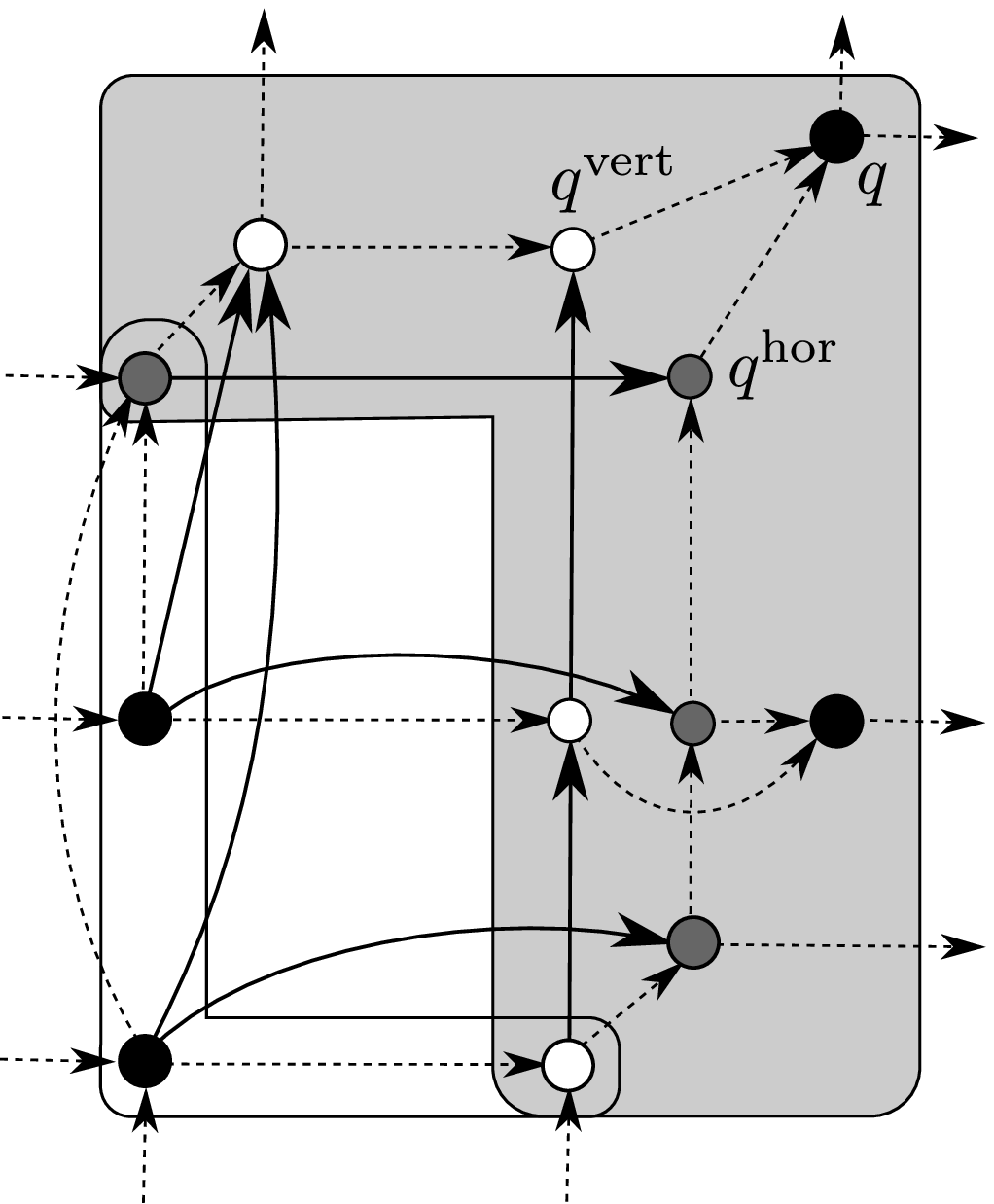}
    \subcaption{}
\end{minipage}
\end{tabular}
\caption{(a) Simplification for a regular pair. (b) Simplification for a flipped pair, where the gray and white vertices distinguish sharing horizontal and vertical segments in $\cH(\GMMN, r) \cap B(l)$, respectively, and the dashed arcs are of length $0$.}
\label{fig:simplified}
\end{figure}

Fix a nondegenerate leaf pair $l \in \GMMN - r$, and we modify the relevant part as follows (see Figure~\ref{fig:simplified}).
We first remove (precisely, avoid creating) the arcs $(p, q) \in E_\mathrm{int}(l)$ for $p \in V_{\llcorner}(r, l)$ and $q \in V_{\urcorner}(r, l)$ with either $p_x < q_x$ and $p_y < q_y$ (diagonal) or $p \in \Vdd(r, l)$.

If $l$ is a regular pair, then we keep the boundary vertices as they are.
Instead of the removed arcs, we add an boundary arc $(q_1, q_2)$ of length $\|q_1q_2\|$ for each $q_1, q_2 \in V_{\urcorner}(r, l)$ with $\{q_1, q_2\} \in E(\mathcal{H}(\GMMN, r))$ and $q_1 \leq q_2$.
Then, for any removed arc $e = (p, q) \in E_\mathrm{int}(l)$, there exists a $p$--$q$ path in $G'$, whose length is always equal to $\ell(e) = \gamma(l, p, q) = d(p, q)$ (cf.~Lemma~\ref{lem:sharable_segments}).

If $l$ is a flipped pair, then we need to care which directional (horizontal or vertical) segments are shared in $B(l)$.
For this purpose, we add two copies $q^\mathrm{hor}$ and $q^\mathrm{vert}$ of each boundary vertex $q \in V_{\urcorner}(r, l) \setminus \Vdd(r, l)$ with two arcs $(q^\mathrm{hor}, q)$ and $(q^\mathrm{vert}, q)$ of length $0$ (recall that, for each $p \in \Vdd(r, l)$, we have already added $p^\mathrm{hor}$ and $p^\mathrm{vert}$, and removed $p$ itself in $G$).
We also replace each remaining axis-aligned arc $(p, q) \in E_\mathrm{int}(l)$ with two arcs $(p, q^\mathrm{hor})$ of length $d_x(p, q)$ and $(p, q^\mathrm{vert})$ of length $d_y(p, q)$.\footnote{We may have $p \in V_\llcorner \setminus \Vdd$ and $q \in \Vdd$, and then the original arc in $G$ is already $(p, q^\mathrm{hor})$ or $(p, q^\mathrm{vert})$ of length $d(p, q) = \|pq\|$. Such an arc is replaced with two arcs $(p, q^\mathrm{hor})$ of length $d_x(p, q)$ and $(p, q^\mathrm{vert})$ of length $d_y(p, q)$.}
Instead of the removed diagonal arcs, we add two boundary arcs $(q_1^\mathrm{hor}, q_2^\mathrm{hor})$ of length $d_x(q_1, q_2)$ and $(q_1^\mathrm{vert}, q_2^\mathrm{vert})$ of length $d_y(q_1, q_2)$ for each $q_1, q_2 \in V_{\urcorner}(r, l)$ with $\{q_1, q_2\} \in E(\mathcal{H}(\GMMN, r))$ and $q_1 \leq q_2$.
Then, for any removed arc $e = (p, q) \in E_\mathrm{int}(l)$, there exist two $p$--$q$ path in $G'$, whose lengths are equal to $d_x(p, q)$ and $d_y(p, q)$.
As $\ell(e) = \gamma(l, p, q) = \max\left\{d_x(p, q),\, d_y(p, q)\right\}$ (cf.~Lemma~\ref{lem:sharable_segments}), the longest paths are preserved by this simplification.

As with Lemma~\ref{lem:star_dag}, we can easily confirm that $G'$ is acyclic.
Thus, we have obtained a simplified auxiliary DAG $G'$, 
and the following lemma completes the proof of Theorem~\ref{thm:star}.

\begin{lemma}
$|V(G')| = O(n^2)$ and $|E(G')| = O(n^2)$.
\end{lemma}

\begin{proof}
For the vertex set, by definition, we see $|V(G')| \leq 3|V(G)| \leq 6|V(\cH(\GMMN, r))| = O(n^2)$.
For the arc set, since all the arcs outside of $\bigcup_{l \in \GMMN - r} B(l)$ directly come from the subgrid $\cH(\GMMN, r)$, it suffices to show that the number of axis-aligned interior arcs and additional boundary arcs is bounded by $O(n^2)$ in total.
By definition, if $\mathcal{H}(\GMMN, r) \cap \mathcal{H}(\GMMN, l)$ is an $a \times b$ grid, then the number of such arcs is at most $3(a + b)$ in the regular case and at most $6(a + b)$ in the flipped case.
Thus, the total number is at most
\begin{align}
\sum_{l \in \GMMN-r} 6|V_{\urcorner}(l)| \leq 6|V(\cH(\GMMN, r))| = O(n^2),
\end{align}
and we are done.
\end{proof}

\section{An $O(n^5)$-Time Algorithm for GMMN[Tree]}\label{sec:tree}
In this section, we present an $O(n^5)$-time algorithm for GMMN[Tree], which is the main target in this paper and stated as follows.

\begin{problem}[{GMMN[Tree]}]
\begin{description}
\setlength{\itemsep}{0mm}
\item[]
\item[Input:] A set $\GMMN \subseteq \mathbb{R}^2 \times \mathbb{R}^2$ of $n$ pairs whose intersection graph $\IG[\GMMN]$ is a tree.
\item[Goal:] Find an optimal network $N = (\pi_v)_{v \in \GMMN} \in \Opt(\GMMN)$.
\end{description}
\end{problem}

For a GMMN[Tree] instance $\GMMN$, we choose an arbitrary pair $r \in \GMMN$ as the root of the tree $\IG[\GMMN]$; in particular, when $\IG[\GMMN]$ is a star, we regard the center as the root.
The basic idea of our algorithm is dynamic programming on the tree $\IG[\GMMN]$ from the leaves toward $r$.
Each subproblem reduces to the longest path problem in DAGs like the star case, which is summarized as follows.

Fix a pair $v = (s_v, t_v) \in \GMMN$.
If $v \neq r$, then there exists a unique parent $u = \mathrm{Par}(v)$ in the tree $\IG[\GMMN]$ rooted at $r$,
and there are $O(n^2)$ possible in-out pairs $(p_v, q_v)$ of $\pi_u \in \Pi_\GMMN(u)$ for $v$.
We virtually define $p_v = q_v = \epsilon$ for the case when we do not care the shared length in $B(u)$, e.g., $v = r$ or $\pi_u$ is disjoint from $B(v)$.
Let $\GMMN_v$ denote the vertex set of the subtree of $\IG[\GMMN]$ rooted at $v$ (including $v$ itself).
For every possible in-out pair $(p_v, q_v)$, as a subproblem, we compute the maximum total length $\mathrm{dp}(v, p_v, q_v)$ of sharable segments in $B(\GMMN_v) = \bigcup_{w \in \GMMN_v}B(w)$, i.e.,
\begin{align}
    \rmdp(v, \epsilon, \epsilon) &= \max\left\{ \sum_{w \in \GMMN_v - v} \| \pi_w \cap \pi_{\mathrm{Par}(w)} \| \Biggm| (\pi_{w})_{w \in \GMMN} \in \Feas(\GMMN) \right\},\\
    \mathrm{dp}(v, p_v, q_v) &= \max\left\{ \sum_{w \in \GMMN_v} \| \pi_w \cap \pi_{\mathrm{Par}(w)} \| \Biggm| (\pi_{w})_{w \in \GMMN} \in \Feas(\GMMN),~\pi_u[v] \in \Pi_\GMMN(p_v, q_v) \right\}. 
\end{align}

By definition, the goal is to compute $\mathrm{dp}(r, \epsilon, \epsilon)$.
If $v$ is a leaf in $\IG[\GMMN]$, then $\GMMN_v = \{v\}$.
In this case, $\mathrm{dp}(v, p_v, q_v)$ is the maximum length of segments shared by two M-paths $\pi_v \in \Pi_\GMMN(v)$ and $\pi_u \in \Pi_\GMMN(u)$ with $\pi_u[v] \in \Pi_\GMMN(p_v, q_v)$, which is easily determined (cf.~Lemma~\ref{lem:sharable_segments}).
Otherwise, using the computed values $\mathrm{dp}(w, p_w, q_w)$ for all children $w$ of $v$ and all possible in-out pairs $(p_w, q_w)$, we reduce the task to the computation of a longest $s_v$--$t_v$ path in an auxiliary DAG, as with finding an optimal M-path for the center pair in the star case.

\subsection{Constructing Auxiliary DAGs for Subproblems}\label{sec:DAG_tree}
Let $v = (s_v, t_v) \in \GMMN$, which is assumed to be regular without loss of generality.
If $v = r$, then let $p_v = q_v = \epsilon$; otherwise, let $u = \mathrm{Par}(v)$ be its parent, and fix a possible in-out pair $u' = (p_v, q_v)$ of $\pi_u \in \Pi_\GMMN(u)$ for $v$ (including the case $p_v = q_v = \epsilon$).
Let $C_v \subseteq \GMMN_v$ be the set of all children of $v$.
By replacing $r$ and $\GMMN - r$ in Section~\ref{sec:reduction} with $v$ and $C_v + u'$ (or $C_v$ if $p_v = q_v = \epsilon$), respectively, we construct the same auxiliary directed graph, denoted by $G[v, p_v, q_v]$.
We then change the length of each interior arc $(p_w, q_w) \in E_\mathrm{int}(w)$ for each child $w \in C_v$ from $\gamma(w, p_w, q_w)$ to $\rmdp(w, p_w, q_w) - \rmdp(w, \epsilon, \epsilon)$, so that it represents the difference of the total sharable length in $B(\GMMN_w) = \bigcup_{w \in \GMMN_w}B(w')$ between the cases when an M-path for $v$ intersects $B(w)$ (enters at $p_w$ and leaves at $q_w$) and when an M-path for $v$ is ignored. 
As with Lemma \ref{lem:star_dag}, the graph $G[v, p_v, q_v]$ is acyclic.
The following lemma completes the reduction of computing $\mathrm{dp}(v, p_v, q_v)$ to finding a longest $s_v$--$t_v$ path in $G[v, p_v, q_v]$. 

\begin{lemma}\label{lem:tree_the_longest}
Let $\pi_G^*$ be a longest $s_v$--$t_v$ path in $G[v, p_v, q_v]$ with respect to $\ell$. We then have
\begin{align}
    \mathrm{dp}(v, p_v, q_v) = \sum_{e \in E(\pi_G^*)} \ell(e) + \sum_{w \in C_v} \mathrm{dp}(w, \epsilon, \epsilon). \label{eq:opt_dag_longest}
\end{align}
\end{lemma}

\begin{proof}
If $v$ is a leaf in $\IG[\GMMN]$, then $C_v = \emptyset$, and hence it immediately follows from Lemma \ref{lem:star_the_longest}.

Suppose that $v$ is not a leaf in $\IG[\GMMN]$, and let $\pi_G$ be a directed $s_v$--$t_v$ path in $G[v, p_v, q_v]$.
We show that there exists a feasible solution $(\pi_w)_{w \in \GMMN} \in \Feas(\GMMN)$ with $\pi_u[v] \in \Pi_\GMMN(p_v, q_v)$ and
\begin{align}
    \sum_{w \in \GMMN_v} \| \pi_{w} \cap \pi_{\mathrm{Par}(w)} \| = \sum_{e \in E(\pi_G)} \ell(e) + \sum_{w \in C_v} \mathrm{dp}(w, \epsilon, \epsilon). \label{eq:pi_G_dp}
\end{align}
By definition, for each $w \in C_v + u'$, the path $\pi_G$ uses at most one arc in $E_\mathrm{int}(w)$.
For each $w \in C_v$ with $E(\pi_G) \cap E_\mathrm{int}(w) \neq \emptyset$, let $e_w = (p_w, q_w)$ be the unique arc in $E(\pi_G) \cap E_\mathrm{int}(w)$, and then
  $\ell(e_w) = \mathrm{dp}(w, p_w, q_w) - \mathrm{dp}(w, \epsilon, \epsilon)$.
Hence, by defining $p_w = q_w = \epsilon$ for each $w \in C_v$ with $E(\pi_G) \cap E_\mathrm{int}(w) = \emptyset$, the right-hand side of \eqref{eq:pi_G_dp} is rewritten as
\begin{align}
  \sum_{w \in C_v} \mathrm{dp}(w, p_w, q_w) + \gamma(u'),
\end{align}
where $\gamma(u') = \gamma(u', p_{u'}, q_{u'})$ if there exists $(p_{u'}, q_{u'}) \in E(\pi_G) \cap E_\mathrm{int}(u')$ and $\gamma(u') = 0$ otherwise.

By the definition of $\mathrm{dp}$, for each $w \in C_v$, there exists an M-path $\tilde\pi_{v, w} \in \Pi_\GMMN(p_w, q_w)$ appearing as $\pi_v[w] = \pi_v \cap \cH(\GMMN, w)$ in some feasible solution $N = (\pi_w)_{w \in \GMMN} \in \Feas(\GMMN)$ such that
\begin{align}
    \sum_{w' \in \GMMN_w} \| \pi_{w'} \cap \pi_{\mathrm{Par}(w')} \| = \mathrm{dp}(w, p_w, q_w). 
\end{align}
If $\gamma(u') = 0$, then $N$ (with replacing $\pi_u$ so that $\|\pi_v \cap \pi_u\| = 0$ if necessary) is a desired network.
Otherwise, there exists a unique arc $(p_{u'}, q_{u'}) \in E(\pi_G) \cap E_\mathrm{int}(u')$.
By choosing $\tilde\pi_{v, u'} \in \Pi_\GMMN(p_{u'}, q_{u'})$ appropriately (cf.~Lemma~\ref{lem:sharable_segments}), one can replace $\pi_v$ as well as $\pi_u$ so that $\|\pi_v \cap \pi_u\| = \gamma(u', p_{u'}, q_{u'})$ and $\pi_u[v] \in \Pi_\GMMN(p_v, q_v)$, and we are done.

To the contrary, we show that, for any feasible solution $N = (\pi_w)_{w \in \GMMN} \in \Feas(\GMMN)$ with $\pi_u[v] \in \Pi_\GMMN(p_v, q_v)$, there exists a directed $s_v$--$t_v$ path $\pi_G$ in $G[v, p_v, q_v]$ of length at least
\begin{align}
  \sum_{w \in \GMMN_v} \| \pi_w \cap \pi_{\mathrm{Par}(w)} \| - \sum_{w \in C_v} \mathrm{dp}(w, \epsilon, \epsilon)
  &= \sum_{w \in C_v} \left(\sum_{w' \in \GMMN_w} \| \pi_{w'} \cap \pi_{\mathrm{Par}(w')} \| - \mathrm{dp}(w, \epsilon, \epsilon)\right) + \|\pi_v \cap \pi_u\|.\label{eq:N_dp}
\end{align}
The proof is done by induction from the leaves to the root in $\IG[\GMMN]$.
For each $w \in C_v + u'$, suppose that $\pi_v[w] \in \Pi_\GMMN(p_w, q_w)$, where we virtually define $p_w = q_w = \epsilon$ if $\pi_v$ is disjoint from $B(w)$.
Then, by taking $\pi_G$ so that $(p_w, q_w) \in E(\pi_G)$ for each $w \in C_v + u'$ unless $p_w = q_w = \epsilon$,
we obtain the following relation from the induction hypothesis (when $v$ is not a leaf) and the definitions of $\ell$ and $\mathrm{dp}$:
\begin{align}
\sum_{e \in E(\pi_G)} \ell(e) = \sum_{w \in C_v} \left(\mathrm{dp}(w, p_w, q_w) - \mathrm{dp}(w, \epsilon, \epsilon)\right) + \gamma(u') \ge \mathrm{(R.H.S.~of~\eqref{eq:N_dp})},
\end{align}
where $\gamma(u') = 0$ if $p_{u'} = q_{u'} = \epsilon$ and $\gamma(u') = \gamma(u', p_{u'}, q_{u'})$ otherwise.
Thus we are done.
\end{proof}

\subsection{Computational Time Analysis}\label{sec:time_tree}
This section completes the proof of Theorem~\ref{thm:tree}.
For a pair $v \in \GMMN$, suppose that $\cH(\GMMN, v)$ is an $a_v \times b_v$ grid graph.
For each possible in-out pair $(p_v, q_v)$, to compute $\mathrm{dp}(v, p_v, q_v)$, we find a longest path in the DAG $G[v, p_v, q_v]$ constructed in Section~\ref{sec:DAG_tree}, which has $O(a_v b_v) = O(n^2)$ vertices and $O(\delta_v (a_v + b_v)^2) = O(\delta_v n^2)$ edges, where $\delta_v$ is the degree of $v$ in $\IG[\GMMN]$.
Hence, for solving the longest path problem once for each $v \in \GMMN$, it takes $\sum_{v \in \GMMN} O(\delta_v n^2) = O(n^3)$ time in total (recall that $\IG[\GMMN]$ is a tree).
For each $v \in \GMMN - r$, there are respectively at most $a_v + b_v = O(n)$ candidates for $p_v$ and for $q_v$, and hence $O(n^2)$ possible in-out pairs.
Thus, the total computational time is bounded by $O(n^5)$, and we are done.
\section{An $O(n^3)$-Time Algorithm for GMMN[Tree]}\label{sec:speedup}
In this section, we improve the DP algorithm for GMMN[Tree] given in Section~\ref{sec:tree} so that it can be implemented in $O(n^3)$ time.

\subsection{Overview}
Let $\GMMN$ be a GMMN[Tree] instance with $|\GMMN| \ge 3$, and we choose a root $r \in \GMMN$ of the tree $\IG[\GMMN]$ such that $r$ is not a leaf (i.e., $r$ has at least two neighbors).
In Section~\ref{sec:tree}, for each $v \in \GMMN - r$ and each possible in-out pair $(p_v, q_v)$ of $\pi_u\in \Pi_\GMMN(u)$ for $v$, we compute $\rmdp(v, p_v, q_v)$ one-by-one by finding a longest $s_v$--$t_v$ path in the auxiliary DAG $G[v, p_v, q_v]$.
In this section, using an extra DP, we improve this part so that we compute $\rmdp(v, p_v, q_v)$ for many possible in-out pairs $(p_v, q_v)$ at once.

As with Section~\ref{sec:tree}, we assume that $v$ is regular, and let $u=(s_u, t_u)$ be the parent of $v$.
We also assume that neither $u$ nor $v$ is degenerate (otherwise, we can easily fill up the table $\rmdp(v, \cdot, \cdot)$ in $O(n^2)$ time by definition).
Since $u$ must have a neighbor other than $v$ by the choice of the root $r$, we have $B(u) \not\subseteq B(v)$.
Hence, for any M-path $\pi_u \in \Pi_\GMMN(u)$, its in-out pair $(p_v, q_v)$ satisfies one of the following conditions:
\begin{enumerate}[label=(\alph*),labelindent=\parindent,leftmargin=*]
    \item either $p_v = s_u \in B(v)$ or $q_v = t_u \in B(v)$, and then it is completely fixed;
    \item $p_v \neq s_u$, $q_v \neq t_u$, and they are on two adjacent boundaries of $B(v)$;
    \item $p_v \neq s_u$, $q_v \neq t_u$, and they are on two opposite boundaries of $B(v)$.
\end{enumerate}
For each case among (a)--(c), we design an extra DP to compute $\rmdp(v, p_v, q_v)$ for all such in-out pairs $(p_v, q_v)$ in $O(n^2)$ time.
Then, no matter how $B(u)$ intersects $B(v)$, one can classify all the possible in-out pairs into a constant number of such cases, and fill up the table $\rmdp(v, \cdot, \cdot)$ in $O(n^2)$ time in total by applying the designed DPs separately.\footnote{If $p_v$ or $q_v$ (or both) can move on two boundaries of $B(v)$, then we separately handle all possible cases, e.g., when $(q_v)_x = (t_v)_x$ (i.e., $q_v$ moves on the right boundary of $B(v)$) and when $(q_v)_y = (t_v)_y$ (i.e., $q_v$ moves on the upper boundary of $B(v)$).}
This implies that the overall computational time is bounded by $O(n^3)$.

No matter which of the three cases (a)--(c) we consider, we first compute the value $\rmdp(v, \epsilon, \epsilon)$ by computing a longest $s_v$--$t_v$ path in the auxiliary DAG $G[v, \epsilon, \epsilon]$.
In addition, by doing it in two ways from $s_v$ and from $t_v$, we obtain a longest $s_v$--$z$ path and a longest $z$--$t_v$ path for every (reachable) $z \in V(G[v, \epsilon, \epsilon])$ as byproducts.
We denote the lengths of the $s_v$--$z$ path and the $z$--$t_v$ path by $\lpath(s_v, z)$ and $\lpath(z, t_v)$, respectively.
Note that this computation for all $v \in P$ requires $O(n^3)$ time in total (cf.~Section~\ref{sec:time_tree}).
We also compute the value $\kappa_v = \sum_{w \in C_v} \rmdp(w, \epsilon, \epsilon)$, which is the baseline of the total sharable length in the subtree rooted at $v$ (cf.~Lemma~\ref{lem:tree_the_longest}), where recall that $C_v$ denotes the set of all children of $v$.

We then show that computing the values $\rmdp(v, p_v, q_v)$ for all possible in-out pairs $(p_v, q_v)$ in each case takes $O(n^2)$ time in total.
Suppose that $\cH(\GMMN, v) \cap \cH(\GMMN, u)$ is an $a \times b$ grid graph, where $a$ and $b$ are associated with the $y$- and $x$-coordinates, respectively.
Depending on the cases (a)--(c) and whether the parent $u$ is regular or flipped (hence, we consider six cases), we define auxiliary DP values (e.g., denoted by $\auxdp(v, i, j)$ for $i \in [a]$ and $j \in [b]$), and demonstrate how to compute and use them.

\subsection{When the Parent is Regular}\label{sec:speedup:regular}
In this section, we consider the case that the parent $u$ is a regular pair. 

\subsubsection{Case (a): One Endpoint is Fixed in the Subgrid}
By symmetry, we consider the situation when $p_v = s_u \in B(v)$ and $(q_v)_x = (t_v)_x$ for all possible in-out pairs $(p_v, q_v)$ of $\pi_u \in \Pi_\GMMN(u)$ for $v$. 
We then have $(t_v)_x < (t_u)_x$ and $(t_u)_y \leq (t_v)_y$, and let $p_{i,j}$ be the $(i, j)$ vertex on the $a \times b$ grid $\mathcal{H}(\GMMN, v) \cap \mathcal{H}(\GMMN, u)$ for each $i \in [a]$ and $j \in [b]$, where we define $p_{1,1} = s_u$ (see Figure~\ref{fig:sppedup:b:regular}).
In this case, we need to compute $\rmdp(v, p_{1,1}, p_{i,b})$ for each $i \in [a]$.

\begin{figure}[tb]
    \centering
    \includegraphics[width=0.45\hsize]{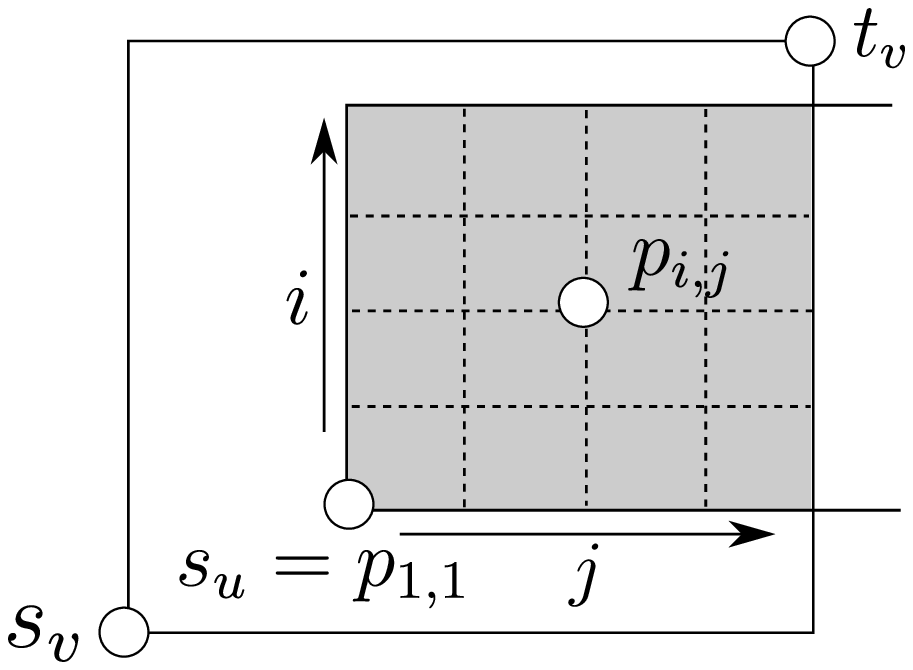}
    \caption{The case (a) when the parent $u$ is regular.}
    \label{fig:sppedup:b:regular}
\end{figure}

For each $i \in [a]$ and $j \in [b]$, we define $\auxdp(v, i, j)$ as the length of a longest $s_v$--$p_{i,j}$ path in $G[v, p_{1,1}, p_{i,j}]$, where we slightly extend the definition of the auxiliary DAG $G[v, p_v, q_v]$ in Section~\ref{sec:DAG_tree} so that $(p_v, q_v)$ is not necessarily an in-out pair of $\pi_u \in \Pi_\GMMN(u)$ for $v$ but that of its subpath.
Then, by Lemma~\ref{lem:tree_the_longest}, we have
\begin{align}
    \rmdp(v, p_{1,1}, p_{i, b}) = \max\left\{\max_{j \in [b]} \left(\omega(v, i, j) + \lambda(p_{i,j}, t_v)\right) + \kappa_v,\, \rmdp(v, \epsilon, \epsilon)\right\}\label{eq:speedup:dp:regular:b}
\end{align}
for each $i \in [a]$, because any $s_v$--$t_v$ path in $G[v, p_{1,1}, p_{i,b}]$ either leaves $B(p_{1,1}, p_{i,b})$ at some $p_{i,j}$ $(j \in [b])$ or is disjoint from $B(p_{1,1}, p_{i,b})$.
Thus, after filling up the table $\omega(v, \cdot, \cdot)$, we can compute the values $\rmdp(v, p_{1,1}, p_{i,b})$ for all $i \in [a]$ in $O(a \times b) = O(n^2)$ time in total.
In what follows, we see how to compute $\omega(v, i, j)$.

For the base case when $i = j = 1$, from the definitions of $G[v, \cdot, \cdot]$ and $\lambda(s_v, \cdot)$, we see
\begin{align}
   \omega(v, 1, 1) &= \lambda(s_v, p_{1,1}). \label{eq:speedup:recurrence:regular:a1}
\end{align}

Next, when $i > 1$ and $j = 1$, we can compute it by a recursive formula
\begin{align+}
   \omega(v, i, 1) &= \max\left\{\omega(v, i-1, 1) + \|p_{i-1, 1}p_{i, 1}\|,\, \lambda(s_v, p_{i,1})\right\},
\end{align+}
which is confirmed as follows.
Fix a longest $s_v$--$p_{i,j}$ path in $G[v, p_{1,1}, p_{i,1}]$ attaining $\omega(v, i, 1)$, and let $\pi \in \Pi_\GMMN(s_v, p_{i,j})$ be a corresponding M-path.
If $\pi$ intersects $p_{i-1,1}$, then the $s_v$--$p_{i-1,1}$ prefix corresponds to a longest $s_v$--$p_{i-1,1}$ path in $G[v, p_{1,1}, p_{i-1,1}]$ of length $\omega(v, i-1, 1)$ and the last segment $p_{i-1,1}p_{i,1}$ contributes to the length in $G[v, p_{1,1}, p_{i,1}]$ in addition.
Otherwise, $\pi$ is disjoint from $p_{i-1,1}$, and it then corresponds to a longest $s_v$--$p_{i,1}$ path in $G[v, \epsilon, \epsilon]$ of length $\lambda(s_v, p_{i,1})$.
The case when $i = 1$ and $j > 1$ is similarly computed by
\begin{align+}
   \omega(v, 1, j) &= \max\left\{\omega(v, 1, j-1) + \|p_{1, j-1}p_{1, j}\|,\, \lambda(s_v, p_{1,j})\right\}.
\end{align+}

Finally, when $i > 1$ and $j > 1$, we can compute it by a recursive formula
\begin{align}
   \omega(v, i, j) &= \max\left\{\omega(v, i-1, j) + \|p_{i-1, j}p_{i, j}\|,\, \omega(v, i, j-1) + \|p_{i,j-1}p_{i,j}\|\right\}, \label{eq:speedup:recurrence:regular:a4}
\end{align}
because for any $s_v$--$p_{i,j}$ path in $G[v, p_{1,1}, p_{i,j}]$, a corresponding M-path in $\Pi_\GMMN(s_v, p_{i,j})$ intersects either $p_{i-1,j}$ or $p_{i,j-1}$, and the last segment $p_{i-1,j}p_{i, j}$ or $p_{i,j-1}p_{i, j}$, respectively, contributes to the length in $G[v, p_{1,1}, p_{i,j}]$.

Since we only look up a constant number of values in \eqref{eq:speedup:recurrence:regular:a1}--\eqref{eq:speedup:recurrence:regular:a4}, each value $\auxdp(v, i, j)$ can be computed in constant time.
As the table $\auxdp(v, \cdot, \cdot)$ is of size $a\times b = O(n^2)$, the total computational time is $O(n^2)$.
Thus we are done.

\subsubsection{Case (b): In-Out Pairs Move on Adjacent Boundaries}
By symmetry, we consider the situation when $(p_v)_y = (s_v)_y$ and $(q_v)_x = (t_v)_x$ for all possible in-out pairs $(p_v, q_v)$ of $\pi_u \in \Pi_\GMMN(u)$ for $v$.
We then have $(s_v)_x \leq (s_u)_x \leq (t_v)_x < (t_u)_x$ and $(s_u)_y < (s_v)_y \leq (t_u)_y \leq (t_v)_y$, and let $p_{i,j}$ be the $(i, j)$ vertex on the $a \times b$ grid $\mathcal{H}(\GMMN, v) \cap \mathcal{H}(\GMMN, u)$ for each $i \in [a]$ and $j \in [b]$, where we define $p_{1,1}$ as the lower-right corner (see Figure~\ref{fig:sppedup:c:regular}).
In this case, we need to compute $\rmdp(v, p_{1,j}, p_{i,1})$ for each pair of $i \in [a]$ and $j \in [b]$.

\begin{figure}[tb]
    \centering
    \includegraphics[width=0.45\hsize]{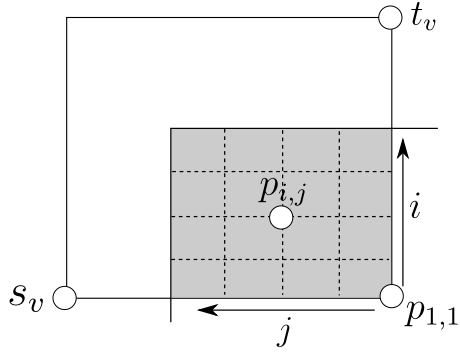}
    \caption{The case (b) when the parent $u$ is regular.}
    \label{fig:sppedup:c:regular}
\end{figure}

For each $i \in [a]$ and $j \in [b]$, we define $\auxdp(v, i, j)$ as the maximum length of an $s_v$--$t_v$ path in $G[v, p_{1,j}, p_{i,1}]$ that intersects $B(p_{1,j}, p_{i,1})$.
Then, by Lemma~\ref{lem:tree_the_longest}, we have
\begin{align}
    \rmdp(v, p_{1,j}, p_{i,1}) &= \max\left\{\auxdp(v, i, j) + \kappa_v,\, \rmdp(v, \epsilon, \epsilon)\right\}.\label{eq:speedup:dp:regular:c}
\end{align}
Thus, after filling up the table $\omega(v, \cdot, \cdot)$, we can compute the values $\rmdp(v, p_{1,j}, p_{i,1})$ for all $i \in [a]$ and $j \in [b]$ in $O(a \times b) = O(n^2)$ time in total. 

In what follows, we see how to compute $\omega(v, i, j)$.
We first observe that, for any $s_v$--$t_v$ path in $G[v, p_{1,j}, p_{i,1}]$ attaining $\omega(v, i, j)$, a corresponding M-path $\pi_v \in \Pi_\GMMN(v)$ can be taken so that it intersects $p_{i,j}$ by choosing an M-path $\pi_u \in \Pi_\GMMN(u)$ appropriately (cf.~Lemma~\ref{lem:sharable_segments}).

For the base case when $i = j = 1$, from the definitions of $G[v, \cdot, \cdot]$, $\lambda(s_v, \cdot)$, and $\lambda(\cdot, t_v)$, we see
\begin{align}
   \omega(v, 1, 1) &= \lambda(s_v, p_{1,1}) + \lambda(p_{1,1}, t_v). \label{eq:speedup:recurrence:regular:c1}
\end{align}

Next, when $i > 1$ and $j = 1$, we can compute it by a recursive formula
\begin{align+}
   \omega(v, i, 1) &= \max\left\{\omega(v, i-1, 1) + \|p_{i-1,1}p_{i,1}\|,\, \lambda(s_v, p_{i,1}) + \lambda(p_{i,1}, t_v)\right\},
\end{align+}
which is confirmed as follows.
Fix an $s_v$--$t_v$ path in $G[v, p_{1,1}, p_{i,1}]$ attaining $\omega(v, i, 1)$, and let $\pi_v \in \Pi_\GMMN(v)$ be a corresponding M-path.
If $\pi_v$ intersects $p_{i-1,1}$, then it corresponds to an $s_v$--$t_v$ path in $G[v, p_{1,1}, p_{i-1,1}]$ attaining $\omega(v, i-1, 1)$ and the segment $p_{i-1,1}p_{i,1}$ contributes to the length in $G[v, p_{1,1}, p_{i,1}]$ in addition.
Otherwise, $\pi_v$ is disjoint from $p_{i-1,1}$, and hence $\pi_v$ intersects $B(p_{1,1}, p_{i,1})$ only at $p_{i,1}$.
Then, the $s_v$--$p_{i,1}$ prefix of $\pi_v$ corresponds to a longest $s_v$--$p_{i,1}$ path in $G[v, \epsilon, \epsilon]$ of length $\lambda(s_v, p_{i,1})$, and the $p_{i,1}$--$t_v$ suffix a longest $p_{i,1}$--$t_v$ path in $G[v, \epsilon, \epsilon]$ of length $\lambda(p_{i,1}, t_v)$.
The case when $i = 1$ and $j > 1$ is similarly computed by
\begin{align+}
   \omega(v, 1, j) &= \max\left\{\omega(v, 1, j-1) + \|p_{1,j}p_{1,j-1}\|,\, \lambda(s_v, p_{1,j}) + \lambda(p_{1,j}, t_v)\right\}.
\end{align+}

Finally, when $i > 1$ and $j > 1$, we can compute it by a recursive formula
\begin{align}
   \omega(v, i, j) &= \max\left\{\omega(v, i-1, j) + \|p_{i-1,j}p_{i,j}\|,\, \omega(v, i, j-1) + \|p_{i,j}p_{i,j-1}\|,\, \lambda(s_v, p_{i,j}) + \lambda(p_{i,j}, t_v)\right\}, \label{eq:speedup:recurrence:regular:c4}
\end{align}
because, for any M-path $\pi_v \in \Pi_\GMMN(v)$ intersecting $p_{i,j}$, it either intersects at least one of $p_{i-1,j}$ and $p_{i,j-1}$ or intersects $B(p_{1,j}, p_{i,1})$ only at $p_{i,j}$, and each case can be analyzed as with the previous paragraph.

Since we only look up a constant number of values in \eqref{eq:speedup:recurrence:regular:c1}--\eqref{eq:speedup:recurrence:regular:c4}, each value $\auxdp(v, i, j)$ can be computed in constant time.
As the table $\auxdp(v, \cdot, \cdot)$ is of size $a\times b = O(n^2)$, the total computational time is $O(n^2)$.
Thus we are done.

\subsubsection{Case (c): In-Out Pairs Move on Opposite Boundaries} 
By symmetry, we consider the situation when $(p_v)_y = (s_v)_y$ and $(q_v)_y = (t_v)_y$ for all possible in-out pairs $(p_v, q_v)$ of $\pi_u \in \Pi_\GMMN(u)$ for $v$.
We then have $(s_v)_x \leq (s_u)_x < (t_u)_x \leq (t_v)_x$ and $(s_u)_y < (s_v)_y < (t_v)_y < (t_u)_y$, and let $p_{i,j}$ be the $(i, j)$ vertex on the $a \times b$ grid $\mathcal{H}(\GMMN, v) \cap \mathcal{H}(\GMMN, u)$ for each $i \in [a]$ and $j \in [b]$, where we define $p_{1,1}$ as the lower-right corner (see Figure~\ref{fig:sppedup:d:regular}). 
In this case, we need to compute $\rmdp(v, p_{1,j}, p_{a,k})$ for each $j, k \in [b]$ with $j \ge k$, which we directly compute as follows.

\begin{figure}[tb]
    \centering
    \includegraphics[width=0.45\hsize]{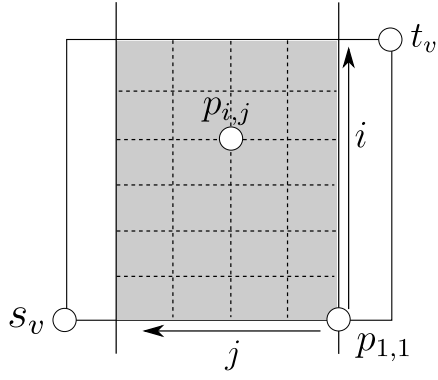}
    \caption{The case (c) when the parent $u$ is regular.}
    \label{fig:sppedup:d:regular}
\end{figure}

First, when $j = k = 1$, we have
\begin{align}
    \rmdp(v, p_{1,1}, p_{a,1}) &= \max_{1 \leq h \leq i \leq a} \left(\lambda(s_v, p_{h,1}) + \lambda(p_{i,1}, t_v) + \|p_{h,1}p_{i,1}\|\right), \label{eq:speedup:dp:regular:d1}
\end{align}
because any M-path $\pi_v \in \Pi_\GMMN(v)$ intersects the segment $p_{1,1}p_{a,1}$ at some point, and it is partitioned into three parts: the $s_v$--$p_{h,1}$ prefix, the segment $p_{h,1}p_{i,1}$, and the $p_{i,1}$--$t_v$ suffix for some $h, i \in [a]$ with $h \leq i$.
The computation of $\rmdp(v, p_{1,1}, p_{a,1})$ requires $O(a^2) = O(n^2)$ time.

Next, for any $1 \leq k \leq j \leq b$, we have
\begin{align}
    \rmdp(v, p_{1,j}, p_{a,k}) &= \rmdp(v, p_{1,1}, p_{a,1}) + \|p_{1,j} p_{1,k}\|, \label{eq:speedup:dp:regular:d2}
\end{align}
which is confirmed as follows.
Fix a network $(\pi_w)_{w \in \GMMN} \in \Feas(\GMMN)$ attaining $\rmdp(v, p_{1,j}, p_{a,k})$.
Then, without changing the total shared length, we can modify the M-paths $\pi_v \in \Pi_\GMMN(v)$ and $\pi_u \in \Pi_\GMMN(u)$ with $\pi_u[v] \in \Pi_\GMMN(p_{1,j}, p_{a,k})$ so that it also attains $\rmdp(v, p_{1,j}, p_{a,j}) = \rmdp(v, p_{1,1}, p_{a,1})$ and $\pi_v$ shares all of its horizontal segments in $B(p_{1,j}, p_{a,k})$ with $\pi_u$ in addition, whose total length is $d_x(p_{1,j}, p_{1,k}) = \|p_{1,j} p_{1,k}\|$ (cf.~Lemma~\ref{lem:sharable_segments} and Figure~\ref{fig:share}).

We can compute $\rmdp(v, p_{1,j}, p_{a,k})$ in constant time by \eqref{eq:speedup:dp:regular:d2} for each $j, k \in [b]$ with $j \geq k$.
As the table $\rmdp(v, \cdot, \cdot)$ is of size $O(b^2) = O(n^2)$, the total computational time is $O(n^2)$.
Thus we are done.

\subsection{When the Parent is Flipped}\label{sec:speedup:flipped}
In this section, we consider the case that the parent $u$ is a flipped pair.

\subsubsection{Case (a): One Endpoint is Fixed in the Subgrid}
By symmetry, we consider the situation when $p_v = s_u \in B(v)$ and $(q_v)_x = (t_v)_x$ for all possible in-out pairs $(p_v, q_v)$ of $\pi_u \in \Pi_\GMMN(u)$ for $v$. 
We then have $(t_v)_x < (t_u)_x$ and $(s_v)_y \leq (t_u)_y$, and let $p_{i,j}$ be the $(i, j)$ vertex on the $a \times b$ grid $\mathcal{H}(\GMMN, v) \cap \mathcal{H}(\GMMN, u)$ for each $i \in [a]$ and $j \in [b]$, where we define $p_{1,1}$ as the upper-right corner so that $p_{1,b} = s_u$ (see Figure~\ref{fig:sppedup:b:flipped}). 
In this case, we need to compute $\rmdp(v, p_{1,b}, p_{i,1})$ for each $i \in [a]$.

\begin{figure}[tb]
    \centering
    \includegraphics[width=0.45\hsize]{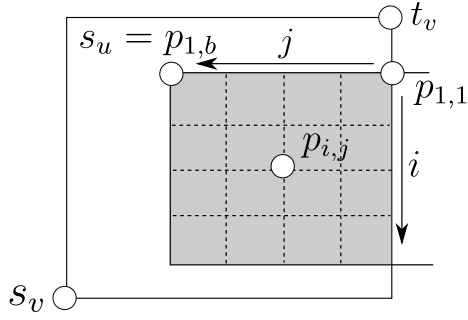}
    \caption{The case (a) when the parent $u$ is flipped.}
    \label{fig:sppedup:b:flipped}
\end{figure}

For each $i \in [a]$ and $j \in [b]$, we define $\auxdp(v, i, j)$ as the length of a longest $p_{i,j}$--$t_v$ path in $G[v, p_{i,j}, p_{1,1}]$.
Then, as with the regular case, by Lemma~\ref{lem:tree_the_longest}, we have
\begin{align}
    &\rmdp(v, p_{1,b}, p_{i,1})\nonumber\\
    &= \max\left\{\max_{j \in [b]} \left(\lambda(s_v, p_{i,j}) + \auxdp(v, i, j)\right) + \kappa_v,\, \max_{h \in [i]} \left(\lambda(s_v, p_{h,b}) + \auxdp(v, h, b)\right) + \kappa_v,\, \rmdp(v, \epsilon, \epsilon)\right\} \label{eq:speedup:dp:flipped:b}
\end{align}
for each $i \in [a]$, because any $s_v$--$t_v$ path in $G[v, p_{i,j}, p_{1,1}]$ either enters $B(p_{1,b}, p_{i,1})$ at some $p_{i,j}$ $(j \in [b])$, enters $B(p_{1,b}, p_{i,1})$ at some $p_{h,b}$ $(h \in [i])$, or is disjoint from $B(p_{1,b}, p_{i,1})$.
Thus, after filling up the table $\omega(v, \cdot, \cdot)$, we can compute the values $\rmdp(v, p_{1,b}, p_{i,1})$ for all $i \in [a]$ in $O(a \times (a + b)) = O(n^2)$ time in total.
In what follows, we see how to compute $\omega(v, i, j)$.

First, when $j = 1$, from the definitions of $G[v, \cdot, \cdot]$ and $\lambda(\cdot, t_v)$, we see
\begin{align}
   \omega(v, i, 1) &= \lambda(p_{1,1}, t_v) + \|p_{i,1}p_{1,1}\|. \label{eq:speedup:recurrence:flipped:b1}
\end{align}
Similarly, when $i = 1$ and $j > 1$, we have
\begin{align}
   \omega(v, 1, j) &= \max_{k \in [j]} \left(\lambda(p_{1,k}, t_v) + \|p_{1,j}p_{1,k}\|\right), \label{eq:speedup:recurrence:flipped:b2}
\end{align}
because any M-path in $\Pi_\GMMN(p_{1,j}, t_v)$ leaves $B(p_{1,j}, p_{1,1})$ at some point $p_{1,k}$ $(k \in [j])$ and then it shares the first segment $p_{1,j}p_{1,k}$ with $p_u \in \Pi_\GMMN(u)$ (with $p_u[v] \in \Pi_\GMMN(s_u, p_{1,1})$).
Computing $\auxdp(v, 1, j)$ requires $O(j)$ time by \eqref{eq:speedup:recurrence:flipped:b2}, and hence it takes $O(b^2) = O(n^2)$ time in total for all $j \in [b]$.

Finally, when $i > 1$ and $j > 1$, we can compute it by a recursive formula
\begin{align}
   \omega(v, i, j) &= \max\left\{\lambda(p_{1,j}, t_v) + \|p_{i,j}p_{1,j}\|,\, \omega(v, i-1, j), \, \omega(v, i, j-1)\right\}, \label{eq:speedup:recurrence:flipped:b3}
\end{align}
which is confirmed as follows.
Fix a longest $p_{i,j}$--$t_v$ path in $G[v, p_{i,j}, p_{1,1}]$ attaining $\omega(v, i, j)$, and let $\pi \in \Pi_\GMMN(p_{i,j}, t_v)$ be a corresponding M-path.
If $\pi$ leaves $B(p_{i,j}, p_{1,1})$ at $p_{1,j}$, then the $p_{1,j}$--$t_v$ suffix corresponds to a longest $p_{1,j}$--$t_v$ path in $G[v, \epsilon, \epsilon]$ of length $\lambda(p_{1,j}, t_v)$ and the first segment $p_{i,j}p_{1,j}$ contributes to the length in $G[v, p_{i,j}, p_{1,1}]$ in addition.
Otherwise, $\pi$ leaves $B(p_{i,j}, p_{1,1})$ at some $p_{1,k}$ $(k \in [j-1])$.
Recall that, since $u$ is flipped, $\pi$ can share either horizontal or vertical segments with $\pi_u \in \Pi_\GMMN(u)$ (cf.~Lemma~\ref{lem:sharable_segments}).
If $\pi$ shares horizontal segments with $\pi_u$, then we can assume that the $p_{i,j}$--$p_{1,k}$ prefix of $\pi$ consists of two segments $p_{i,j}p_{1,j}$ and $p_{1,j}p_{1,k}$ by modifying $\pi_u$ (with $\pi_u[v] \in \Pi_\GMMN(p_{1,b}, p_{i,1})$) so that it traverses $p_{1,j}p_{1,k}$;
we then have $\omega(v, i, j) = \omega(v, 1, j) = \omega(v, i-1, j)$.
Otherwise, $\pi$ shares vertical segments with $\pi_u$, and we can assume that the $p_{i,j}$--$p_{1,k}$ prefix of $\pi$ consists of two segments $p_{i,j}p_{i,k}$ and $p_{i,k}p_{1,k}$ by modifying $\pi_u$ so that it traverses $p_{i,k}p_{1,k}$;
we then have $\omega(v, i, j) = \omega(v, i, k) = \omega(v, i, j-1)$.

Since we only look up a constant number of values in \eqref{eq:speedup:recurrence:flipped:b3} as well as \eqref{eq:speedup:recurrence:flipped:b1}, each value $\auxdp(v, i, j)$ for $i > 1$ can be computed in constant time.
As the table $\auxdp(v, \cdot, \cdot)$ is of size $a\times b = O(n^2)$, the total computational time is bounded by $O(n^2)$.
Thus we are done.

\subsubsection{Case (b): In-Out Pairs Move on Adjacent Boundaries}
By symmetry, we consider the situation when $(p_v)_y = (t_v)_y$ and $(q_v)_x = (t_v)_x$ for all possible in-out pairs $(p_v, q_v)$ of $\pi_u \in \Pi_\GMMN(u)$ for $v$.
We then have $(s_v)_x \leq (s_u)_x \leq (t_v)_x < (t_u)_x$ and $(s_v)_y \leq (t_u)_y \leq (t_v)_y < (s_u)_y$, and let $p_{i,j}$ be the $(i, j)$ vertex on the $a \times b$ grid $\mathcal{H}(\GMMN, v) \cap \mathcal{H}(\GMMN, u)$ for each $i \in [a]$ and $j \in [b]$, where we define $p_{1,1}=t_v$ (see Figure~\ref{fig:sppedup:a:flipped}). 
In this case, we need to compute $\rmdp(v, p_{1,j}, p_{i,1})$ for each pair of $i \in [a]$ and $j \in [b]$.

\begin{figure}[tb]
    \centering
    \includegraphics[width=0.45\hsize]{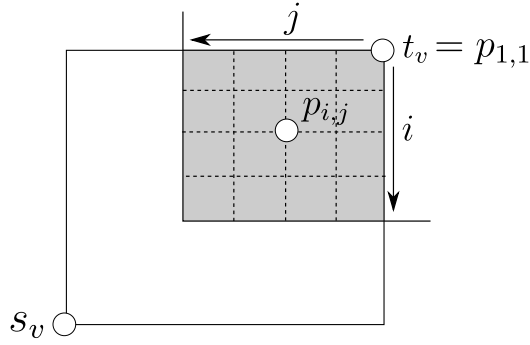}
    \caption{The case (b) when the parent $u$ is flipped.}
    \label{fig:sppedup:a:flipped}
\end{figure}

For each $i \in [a]$ and $j \in [b]$, we define $\auxdp(v, i, j)$ as the length of a longest $s_v$--$t_v$ path in $G[v, p_{1,j}, p_{i,1}]$.
Then, by Lemma~\ref{lem:tree_the_longest}, we have
\begin{align}
    \rmdp(v, p_{1,j}, p_{i,1}) &= \auxdp(v, i, j) + \kappa_v. \label{eq:speedup:dp:flipped:a}
\end{align}
Thus, after filling up the table $\omega(v, \cdot, \cdot)$, we can compute the values $\rmdp(v, p_{1,j}, p_{i,1})$ for all $i \in [a]$ and $j \in [b]$ in $O(a \times b) = O(n^2)$ time in total. 
In what follows, we see how to compute $\omega(v, i, j)$.

For the base case when $i = j = 1$, from the definitions of $G[v, \cdot, \cdot]$ and $\lambda(s_v, \cdot)$, we see
\begin{align}
   \omega(v, 1, 1) &= \lambda(s_v, p_{1,1}). \label{eq:speedup:recurrence:flipped:a1}
\end{align}

Next, when $i > 1$ and $j = 1$, we can compute it by a recursive formula
\begin{align+}
   \omega(v, i, 1) &= \max\left\{\lambda(s_v, p_{i,1}) + \|p_{i,1}p_{1,1}\|,\, \omega(v, i-1, 1)\right\},
\end{align+}
which is confirmed as follows.
Fix a longest $s_v$--$t_v$ path in $G[v, p_{i,1}, p_{1,1}]$ attaining $\omega(v, i, 1)$, and let $\pi_v \in \Pi_\GMMN(v)$ be a corresponding M-path.
If $\pi_v$ intersects $p_{i,1}$, then it corresponds to a longest $s_v$--$p_{i,1}$ path in $G[v, \epsilon, \epsilon]$ of length $\lambda(s_v, p_{i,1})$ and the last segment $p_{i,1}p_{1,1}$ contributes to the length in $G[v, p_{1,1}, p_{i,1}]$ in addition.
Otherwise, $\pi_v$ is disjoint from $p_{i,1}$, and then it corresponds to a longest $s_v$--$t_v$ path in $G[v, p_{i-1,1}, p_{1,1}]$ of length $\omega(v, i-1, 1)$.
The case when $i = 1$ and $j > 1$ is similarly computed by
\begin{align+}
   \omega(v, 1, j) &= \max\left\{\lambda(s_v, p_{1,j}) + \|p_{1,j}p_{1,1}\|,\, \omega(v, 1, j-1)\right\}.
\end{align+}

Finally, when $i > 1$ and $j > 1$, we can compute it by a recursive formula
\begin{align}
   \omega(v, i, j) &= \max\left\{\lambda(s_v, p_{i,j}) + \gamma(v, p_{i,j}, p_{1,1}),\, \omega(v, i-1, j), \, \omega(v, i, j-1)\right\}, \label{eq:speedup:recurrence:flipped:a4}
\end{align}
where $\gamma(v, p_{i,j}, p_{1,1}) = \max\left\{d_x(p_{i,j}, p_{1,1}),\, d_y(p_{i,j}, p_{1,1})\right\}$ is similarly defined (cf.~\eqref{eq:sharable_length} and Lemma~\ref{lem:sharable_segments}).
This is because, for any M-path $\pi_v \in \Pi_\GMMN(v)$, it intersects $p_{i,j}$, enters $B(p_{i-1,j}, p_{1,1})$ at some $p_{h,j}$ $(h \in [i-1])$, or enters $B(p_{i,j-1}, p_{1,1})$ at some $p_{i,k}$ $(k \in [j-1])$, and each case is analyzed as with the previous paragraph.

Since we only look up a constant number of values in \eqref{eq:speedup:recurrence:flipped:a1}--\eqref{eq:speedup:recurrence:flipped:a4}, each value $\auxdp(v, i, j)$ can be computed in constant time.
As the table $\auxdp(v, \cdot, \cdot)$ is of size $a\times b = O(n^2)$, the total computational time is $O(n^2)$.
Thus we are done.

\subsubsection{Case (c): In-Out Pairs Move on Opposite Boundaries}
By symmetry, we consider the situation when $(p_v)_y = (t_v)_y$ and $(q_v)_y = (s_v)_y$ for all possible in-out pairs $(p_v, q_v)$ of $\pi_u \in \Pi_\GMMN(u)$ for $v$.
We then have $(s_v)_x \leq (s_u)_x < (t_u)_x \leq (t_v)_x$ and $(t_u)_y < (s_v)_y < (t_v)_y < (s_u)_y$, and let $p_{i,j}$ be the $(i, j)$ vertex on the $a \times b$ grid $\mathcal{H}(\GMMN, v) \cap \mathcal{H}(\GMMN, u)$ for each $i \in [a]$ and $j \in [b]$, where we define $p_{1,1}$ as the upper-right corner (see Figure~\ref{fig:sppedup:d:flipped}). 
In this case, we need to compute $\rmdp(v, p_{1,j}, p_{a,k})$ for each $j, k \in [b]$ with $j \ge k$.
Recall that, since $u$ is flipped, any M-paths $\pi_v \in \Pi_{\GMMN}(v)$ and $\pi_u \in \Pi_{\GMMN}(u)$ can share either horizontal or vertical segments.

\begin{figure}[tb]
    \centering
    \includegraphics[width=0.45\hsize]{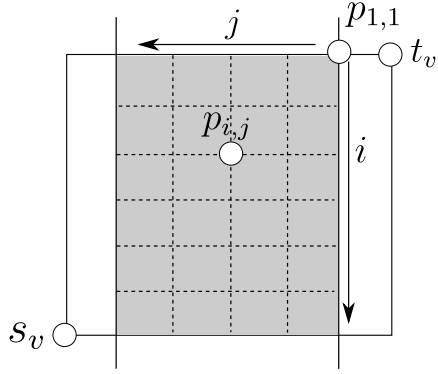}
    \caption{The case (c) when the parent $u$ is flipped.}
    \label{fig:sppedup:d:flipped}
\end{figure}

First, when $j = k$, as with the regular case (cf.~\eqref{eq:speedup:dp:regular:d1}), we have
\begin{align}
    \rmdp(v, p_{1,j}, p_{a,j}) &= \rmdp(v, p_{1,1}, p_{a,1}) = \max_{1 \leq h \leq i \leq a} \left(\lambda(s_v, p_{h,1}) + \lambda(p_{i,1}, t_v) + \|p_{h,1}p_{i,1}\|\right) \label{eq:speedup:dp:flipped:d1}
\end{align}
because in this case $\pi_u[v]$ consists of a single vertical segment $p_{1,j}p_{a,j}$.

When $j > k$, we consider two cases of sharing horizontal and vertical segments separately, and then take the maximum.
In the vertical sharing case, the desired value is exactly $\rmdp(v, p_{1,1}, p_{a,1})$, because any horizontal segment in $B(p_{1,j}, p_{a,k})$ has no meaning.
In the horizontal sharing case, the desired value is $\rmdp(v, \epsilon, \epsilon) + \|p_{1,k}p_{1,j}\|$, because for any longest $s_v$--$t_v$ path in $G[v, \epsilon, \epsilon]$, we can take a corresponding M-path $\pi_v \in \Pi_\GMMN(v)$ so that it goes through $B(p_{1,j}, p_{a,k})$ horizontally and then it can share the horizontal segment in addition with $\pi_u \in \Pi_\GMMN(u)$ (with $\pi_u[v] \in \Pi_\GMMN(p_{1,j}, p_{a,k})$).
Thus, we have
\begin{align}
    \rmdp(v, p_{1,j}, p_{a,k}) &= \max\left\{\rmdp(v, p_{1,1}, p_{a,1}),\, \rmdp(v, \epsilon, \epsilon) + \|p_{1,k}p_{1,j}\|\right\}. \label{eq:speedup:dp:flipped:d2}
\end{align}

The computation of $\rmdp(v, p_{1,1}, p_{a,1})$ requires $O(a^2) = O(n^2)$ time by \eqref{eq:speedup:dp:flipped:d1}.
After computing it, by \eqref{eq:speedup:dp:flipped:d1} and \eqref{eq:speedup:dp:flipped:d2}, we can compute $\rmdp(v, p_{1,j}, p_{a,k})$ in constant time for each $j, k \in [b]$ with $j \geq k$.
As the table $\rmdp(v, \cdot, \cdot)$ is of size $O(b^2) = O(n^2)$, the total computational time is $O(n^2)$.
Thus we are done.
\section{Reduction of GMMN[Cycle] to GMMN[Tree]}\label{sec:cycle}
In this section, we show that GMMN[Cycle] can be reduced to $O(n)$ GMMN[Tree] instances.
More generally, we describe a reduction for triangle-free pseudotree instances.
The target problem is formally stated as follows, where we emphasize again that the triangle-freeness is crucial in our approach (cf.~Section~\ref{sec:specialization}).

\begin{problem}[{GMMN[Pseudotree]}]
\begin{description}
\setlength{\itemsep}{0mm}
\item[]
\item[Input:] A set $\GMMN \subseteq \mathbb{R}^2 \times \mathbb{R}^2$ of $n$ pairs whose intersection graph $\IG[\GMMN]$ is a triangle-free pseudotree.
\item[Goal:] Find an optimal network $N = (\pi_v)_{v \in \GMMN} \in \Opt(\GMMN)$.
\end{description}
\end{problem}

Let $\GMMN$ be a GMMN[Pseudotree] instance.
If $\IG[\GMMN]$ is a tree, we do nothing for reduction.
Suppose that $\IG[\GMMN]$ has a (unique) cycle of length at least four.
Let $C \subseteq P$ be the subset of pairs consisting of the cycle.
If $C$ has a degenerate pair, we can cut the cycle by appropriately splitting the degenerate pair into two degenerate pairs which are not adjacent in the intersection graph.
Therefore, we can assume that any pair in $C$ is not degenerate.

We choose an arbitrary pair $v = (s_v, t_v) \in C$.
Without loss of generality, we assume that $v$ is regular and $s_v$ is the lower-left corner of $B(v)$.
Suppose that $\cH(\GMMN, v)$ is an $a \times b$ grid graph, where $a$ and $b$ are associated with the $y$- and $x$-coordinates, respectively.
Note that $a, b \ge 2$ since $v$ is nondegenerate.
Let $p_{i,j}$ denote the $(i,j)$ vertex for $i \in [a]$ and $j \in [b]$, where $p_{1,1} = s_v$.
For each $i \in [a]$ and $j \in [b]$, we define
\begin{align}
  \mathcal{E}^\mathrm{hor}(p_{i,j}) &= \begin{cases}
    \{(p_{i,j-1}, p_{i,j}, p_{i,j+1}), (p_{i-1,j}, p_{i,j}, p_{i,j+1})\} & (1 < i,\ 1 < j < b), \\
    \{(p_{i,j-1}, p_{i,j}, p_{i,j+1})\} & (i = 1,\ 1 < j < b), \\
    \{(p_{i-1,j}, p_{i,j}, p_{i,j+1})\} & (1 < i,\ j = 1), \\
    \{(p_{i,j}, p_{i,j}, p_{i,j+1})\} & (i = j = 1), \\
    \emptyset & (\text{otherwise, i.e., } j = b),
  \end{cases}
  \\
  \mathcal{E}^\mathrm{vert}(p_{i,j}) &= \begin{cases}
    \{(p_{i-1,j}, p_{i,j}, p_{i+1,j}), (p_{i,j-1}, p_{i,j}, p_{i+1,j})\} & (1 < i < a,\ 1 < j), \\
    \{(p_{i-1,j}, p_{i,j}, p_{i+1,j})\} & (1 < i < a,\ j = 1), \\
    \{(p_{i,j-1}, p_{i,j}, p_{i+1,j})\} & (i = 1,\ 1 < j), \\
    \{(p_{i,j}, p_{i,j}, p_{i+1,j})\} & (i = j = 1), \\
    \emptyset & (\text{otherwise, i.e., } i = a).
  \end{cases}
\end{align}
Namely, each element of $\mathcal{E}^\mathrm{hor}(p_{i,j})$ is a triple representing a way for an M-path $\pi_v \in \Pi_\GMMN(v)$ to go through an edge $\{p_{i,j}, p_{i, j+1}\}$ of $\cH(\GMMN, v)$.
Similarly, each element of $\mathcal{E}^\mathrm{vert}(p_{i,j})$ indicates a manner for $\pi_v$ to go through $\{p_{i,j}, p_{i+1,j}\}$.

Let $u_1$ and $u_2$ be the neighbors of $v$ in $C$.
Then $B(u_1)$ and $B(u_2)$ can be separated by an axis-aligned line, without their boundaries (recall that they can share corner vertices).
By symmetry, we assume that the line is vertical and $B(u_1)$ is the left side.
Take $\alpha \in [a]$ and $\beta \in [b]$ such that $(p_{\alpha, \beta})_x$ is the $x$-coordinate of the right boundary of $B(u_1) \cap B(v)$ and $(p_{\alpha, \beta})_y$ is the minimum of the $y$-coordinates of the upper boundaries of $B(u_1) \cap B(v)$ and $B(u_2) \cap B(v)$.
If $\alpha = a$ and $\beta = b$, i.e., $p_{\alpha,\beta} = t_v$, the lower-right corner of $B(u_1)$ and the upper-left corner of $B(u_2)$ are $t_v$.
In this case, we flip both the $x$- and $y$-axes so that $p_{\alpha,\beta} = s_v$.
Hence, we can assume that $p_{\alpha,\beta} \ne t_v$.

Define
\begin{align}
  X^\mathrm{hor} &= \{p_{i, \beta} \mid i \in [\alpha]\}, \\
  X^\mathrm{vert}  &= \{p_{\alpha, j} \mid j \in [\beta]\}.
\end{align}
Then any M-path $\pi_v$ is consistent with exactly one way in $\mathcal{E}^\mathrm{hor}(q)$ for some $q \in X^\mathrm{hor}$ or in $\mathcal{E}^\mathrm{vert}(q)$ for some $q \in X^\mathrm{vert}$.
We try every possibility and then adopt an optimal one.

Assume that $\pi_v$ is consistent with $(q^-, q, q^+) \in \mathcal{E}^\mathrm{hor}(q)$ for some $q \in X^\mathrm{hor}$ or with $(q^-, q, q^+) \in \mathcal{E}^\mathrm{vert}(q)$ for some $q \in X^\mathrm{vert}$, i.e., $\pi_v$ goes through $\{q^-, q\}$ and $\{q, q^+\}$.
Then the minimum length of a network under this assumption is the same as $\tilde{N} \in \Opt(\tilde{\GMMN})$, where $\tilde{\GMMN} = (\GMMN - v) \cup \{v_1, v_2, v_3, v_4\}$ with $v_1 = (s_v, q^-)$, $v_2 = (q^-, q)$, $v_3 = (q, q^+)$, and $v_4 = (q^+, t_v)$ (see Figure~\ref{fig:cycle1}).
It is shown that $\IG[\tilde{\GMMN}]$ has no cycles as follows.

\begin{figure}[tb]
\begin{tabular}{cc}
\begin{minipage}{0.48\hsize}
    \centering
    \includegraphics[width=\hsize]{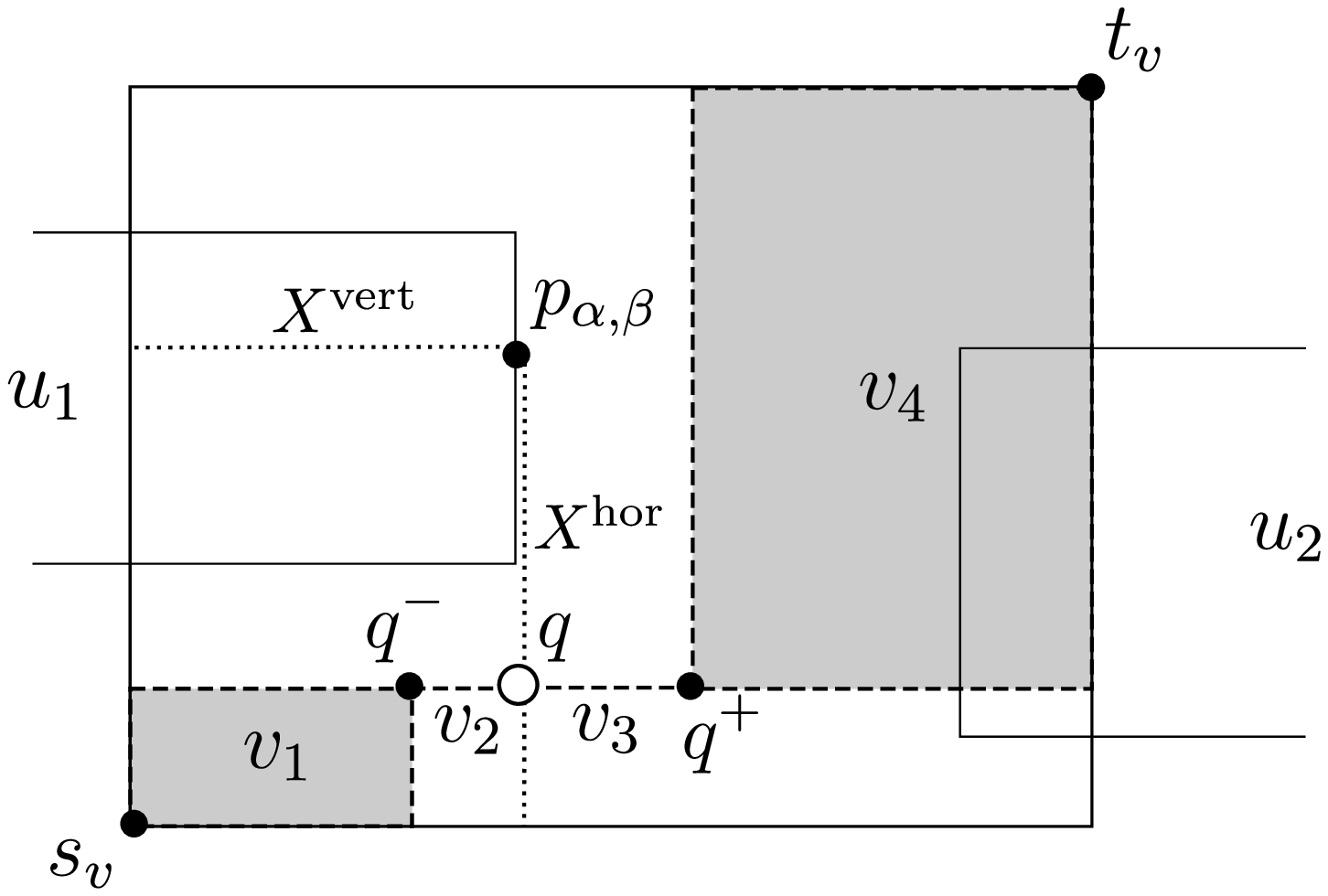}
    \subcaption{}
\end{minipage}&
\begin{minipage}{0.48\hsize}
    \centering
    \includegraphics[width=\hsize]{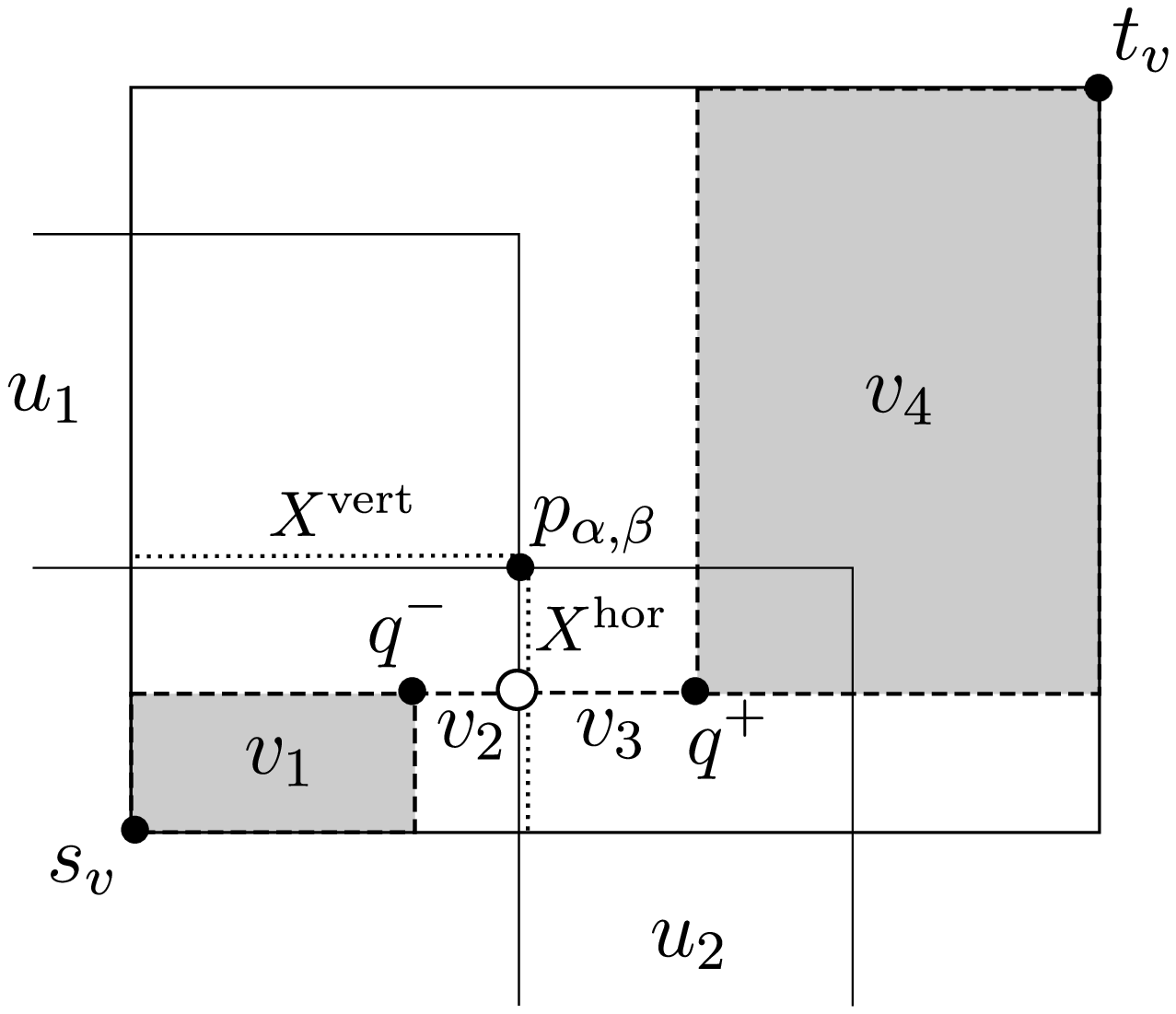}
    \subcaption{}
\end{minipage}
\end{tabular}
\caption{Construction of $\tilde{\GMMN} = (\GMMN - v) \cup \{v_1, v_2, v_3, v_4\}$, where $q \in X^\mathrm{hor} \cup X^\mathrm{vert}$ is on dotted lines.}
\label{fig:cycle1}
\end{figure}

\begin{claim}\label{lem:reduction_forest}
  $\IG[\tilde{\GMMN}]$ is a forest.
\end{claim}

\begin{proof}
  Let $\Gamma_v$ be the set of neighbors of $v$ in $\IG[\GMMN]$, excluding $v$ itself.
  Then $\IG[\Gamma_v]$ is edgeless as $P$ is triangle-free.
  Put $\tilde{\Gamma} = \Gamma_v \cup \tilde{V}$ with $\tilde{V} = \{v_1, v_2, v_3, v_4\}$.
  Then, for every $v_k \in \tilde{V}$, neighbors of $v_k$ in $\IG[\tilde{\GMMN}]$ are included in $\Gamma_v$ by $B(v_k) \subseteq B(v)$.
  In addition, for $k \in \{2, 3\}$, the graph $\cH(\tilde{\GMMN}, v_k)$ consists of a single edge in $E(\cH(\GMMN, v))$ or is a single vertex in $V(\cH(\GMMN, v))$.
  Therefore, there exists at most one pair $w \in \Gamma_v$ such that $B(w)$ intersects $B(v_k)$.
  This means that the degree of $v_2$ and $v_3$ in $\IG[\tilde{\Gamma}]$ is at most one, and hence they are not in any cycle in $\IG[\tilde{\GMMN}]$.
  
  Since the pairs in $\Gamma_v$ are not adjacent to each other in $\IG[\GMMN]$ and we have $q^- \le q^+$ and $q^- \ne q^+$ by definition, at most one pair in $\Gamma_v$ can be adjacent to both $v_1$ and $v_4$ in $\IG[\tilde\Gamma]$.
  Thus $\IG[\tilde\Gamma]$ is a forest.
  Then, if $\IG[\tilde\GMMN]$ has a cycle, at least one of the following holds:
  \begin{enumerate}[label={(\arabic*)},labelindent=\parindent,leftmargin=*]
    \item[(C1)] $\{u_1, v_1\}, \{v_1, u_2\} \in E(\IG[\tilde{\Gamma}])$,
    \item[(C2)] $\{u_1, v_4\}, \{v_4, u_2\} \in E(\IG[\tilde{\Gamma}])$, or
    \item[(C3)] there exists $w \in \Gamma_v \setminus \{u_1, u_2\}$ such that $\{u_1, v_1\}, \{v_1, w\}, \{w, v_4\}, \{v_4, u_2\} \in E(\IG[\tilde{\Gamma}])$.
  \end{enumerate}
  In what follows, we see that none of these is the case.
  
  Let $i, h_1, h_2, g_1, g_2 \in [a]$ and $j, \gamma \in [b]$ such that $q = p_{i,j}$, the upper- and lower-right corners of $B(u_1) \cap B(v)$ are $p_{h_1,\beta}$ and $p_{g_1,\beta}$, respectively, and the upper- and lower-left corners of $B(u_2) \cap B(v)$ are $p_{h_2,\gamma}$ and $p_{g_2,\gamma}$, respectively.
  Note that $\alpha = \min \left\{h_1, h_2\right\}$ and $\beta \leq \gamma$.
  
  We first consider the exceptional case when $q = s_v \ (= p_{1,1})$. 
  Since $\cH(P, v_1)$ consists of the single vertex $s_v$, neither~(C1) nor~(C3) holds.
  If $q^+ = p_{1,2}$, then $p_{1,1}\in X^\mathrm{hor}$ and $\beta = 1$, which implies $B(u_1) \cap B(v_4) = \emptyset$.
  Otherwise (i.e., $q^+ = p_{2,1}$), we have $p_{1,1}\in X^\mathrm{vert}$.
  This implies $\min\left\{h_1, h_2\right\} = \alpha = 1$, and hence $B(u_1) \cap B(v_4) = \emptyset$ or $B(u_2) \cap B(v_4) = \emptyset$.
  Thus~(C2) does not hold, and we are done.
 
  We next deal with the case when $q \in X^{\mathrm{hor}} - s_v$ and $(q^-, q, q^+) \in \mathcal{E}^{\mathrm{hor}}$.
  By the definition of $X^{\mathrm{hor}}$, it holds $j = \beta$.
  Consider~(C1).
  If $\beta < \gamma$, we have $B(u_2) \cap B(v_1) = \emptyset$, which negates~(C1).
  Suppose that $\beta = \gamma$.
  Then it must holds $h_1 \le g_2$ or $h_2 \le g_1$ since otherwise $B(u_1)$ intersect $B(u_2)$.
  In the former case, we have $\alpha = h_1 \le g_2$ and thus $B(v_1)$ does not intersect $B(u_2)$.
  In the latter case, say $\alpha = h_2 \le g_1$ (see Figure~\ref{fig:cycle3}), it holds $B(v_1) \cap B(u_2) = \emptyset$ if $q^- = p_{i, \beta-1}$ and $B(v_1) \cap B(u_1) = \emptyset$ if $q^- = p_{i-1, \beta}$.
  Therefore,~(C1) does not hold for any case.
  We also have $B(u_1) \cap B(v_4) = \emptyset$ by $q_x \le q^+_x$ and $q_x \ne q^+_x$, which means that~(C2) does not hold either.
  
  \begin{figure}[tb]
    \begin{tabular}{cc}
      \begin{minipage}{0.48\hsize}
        \centering
        \includegraphics[width=\hsize]{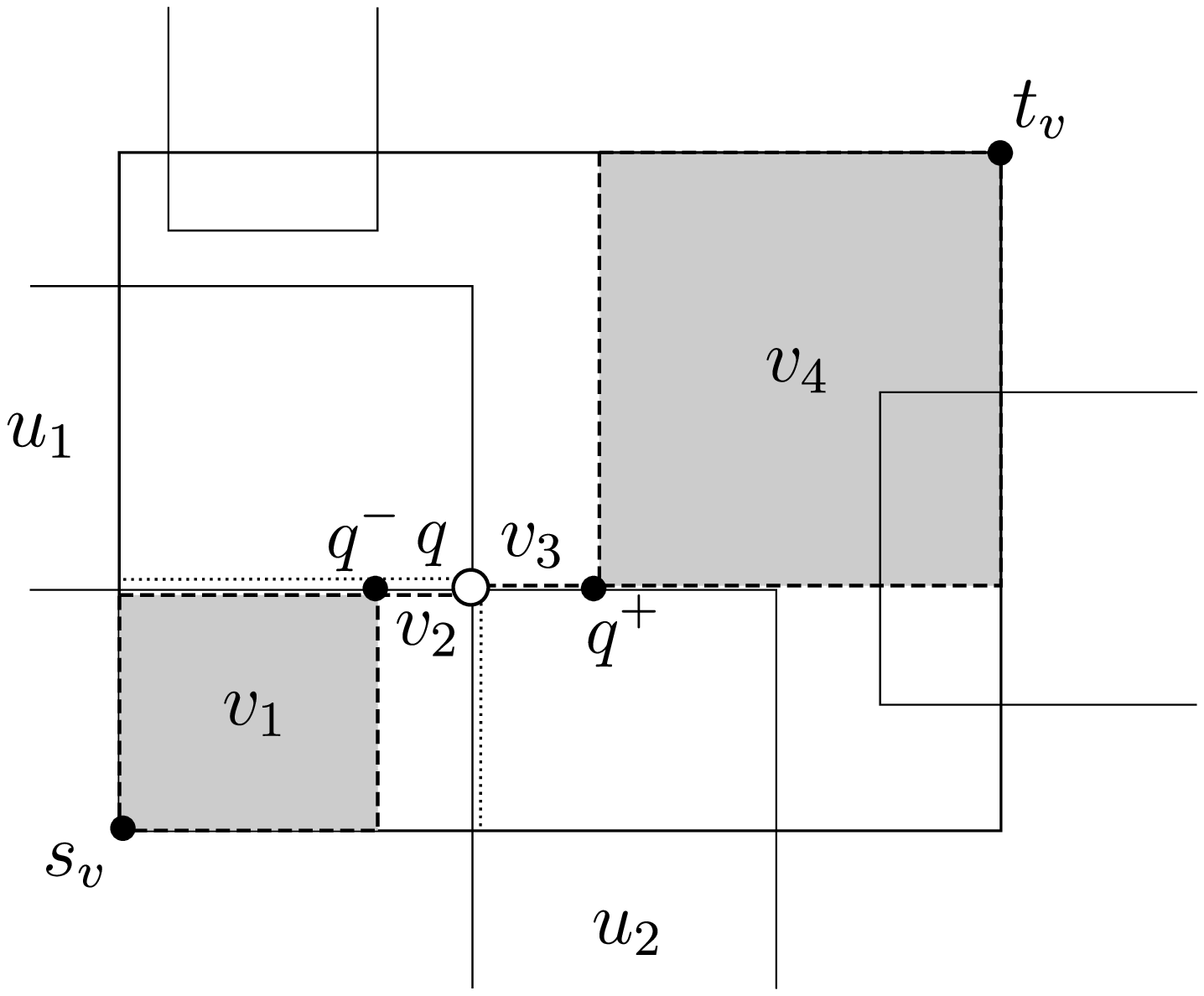}
        \subcaption{}
      \end{minipage}&
      \begin{minipage}{0.48\hsize}
        \centering
        \includegraphics[width=\hsize]{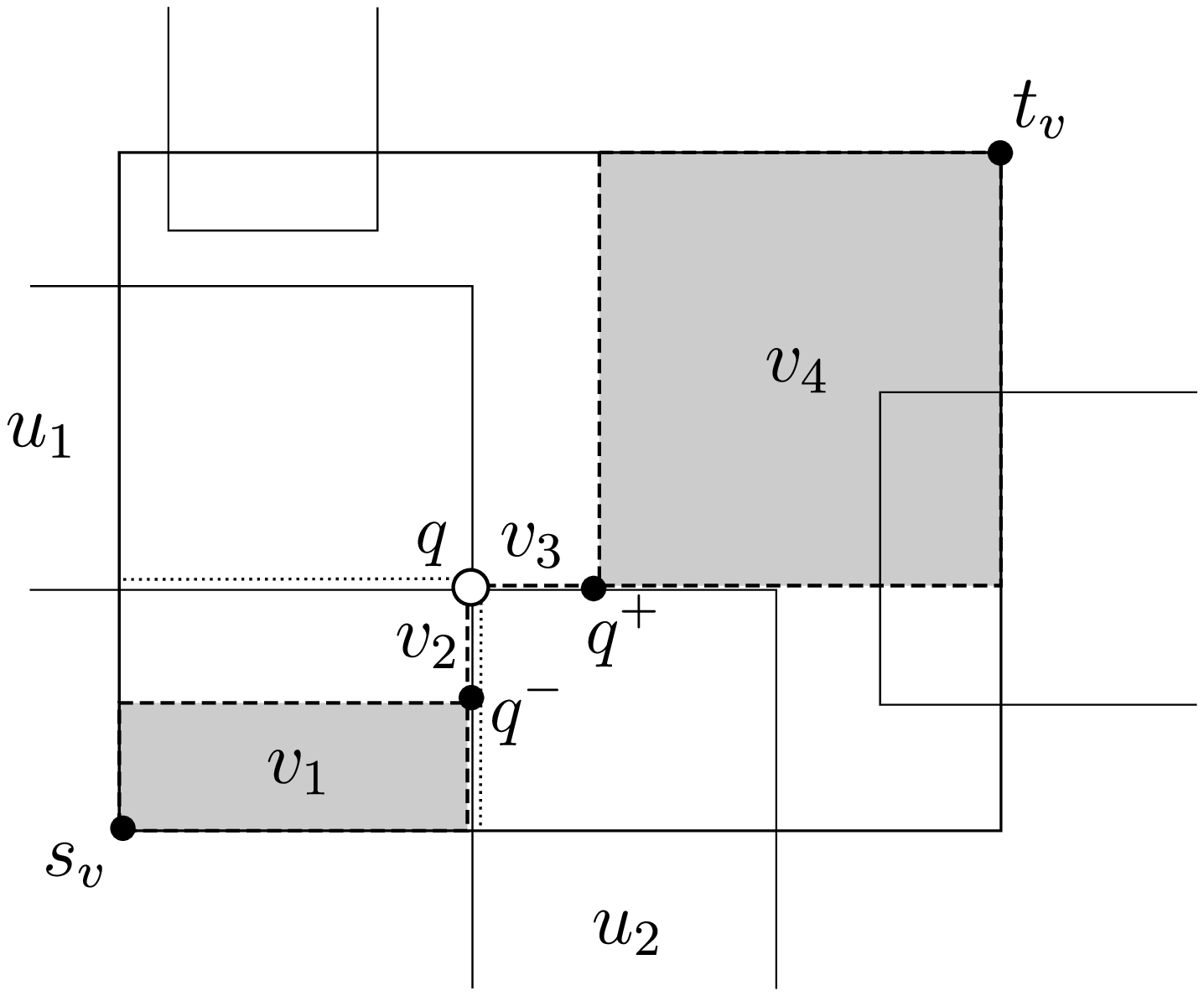}
        \subcaption{}
      \end{minipage}
    \end{tabular}
    \caption{An uneasy situation when $B(u_1)$ and $B(u_2)$ share their corners at $q$ (i.e., $\beta = \gamma = i$ and $g_1 = h_2$). (a) $\pi_v \in \Pi_\GMMN(v)$ goes through $q$. (b) $\pi_v \in \Pi_\GMMN(v)$ turns at $q$.}
    \label{fig:cycle3}
  \end{figure}

  Suppose that there exists $w \in \Gamma_v \setminus \{u_1, u_2\}$ satisfying the condition in~(C3).
  If $i < g_1$, then $B(v_1)$ does not intersect $B(u_1)$, which contradicts $\{u_1, v_1\} \in E(\IG[\tilde{\Gamma}])$.
  If $i > g_1$, then $B(w)$ must intersect $B(u_1)$, which also contradicts the assumption that $P$ is triangle-free.
  Otherwise, $i = g_1$ (see Figure~\ref{fig:cycle2}).
  If $q^- = p_{i, \beta-1}$, then $B(w)$ intersects $B(u_1)$, a contradiction again.
  If $q^- = p_{i-1, \beta}$, the bounding box $B(v_1)$ does not intersect $B(u_1)$.
  Hence~(C3) is not true for any case.
  
  \begin{figure}[tb]
    \begin{tabular}{cc}
      \begin{minipage}{0.48\hsize}
        \centering
        \includegraphics[width=\hsize]{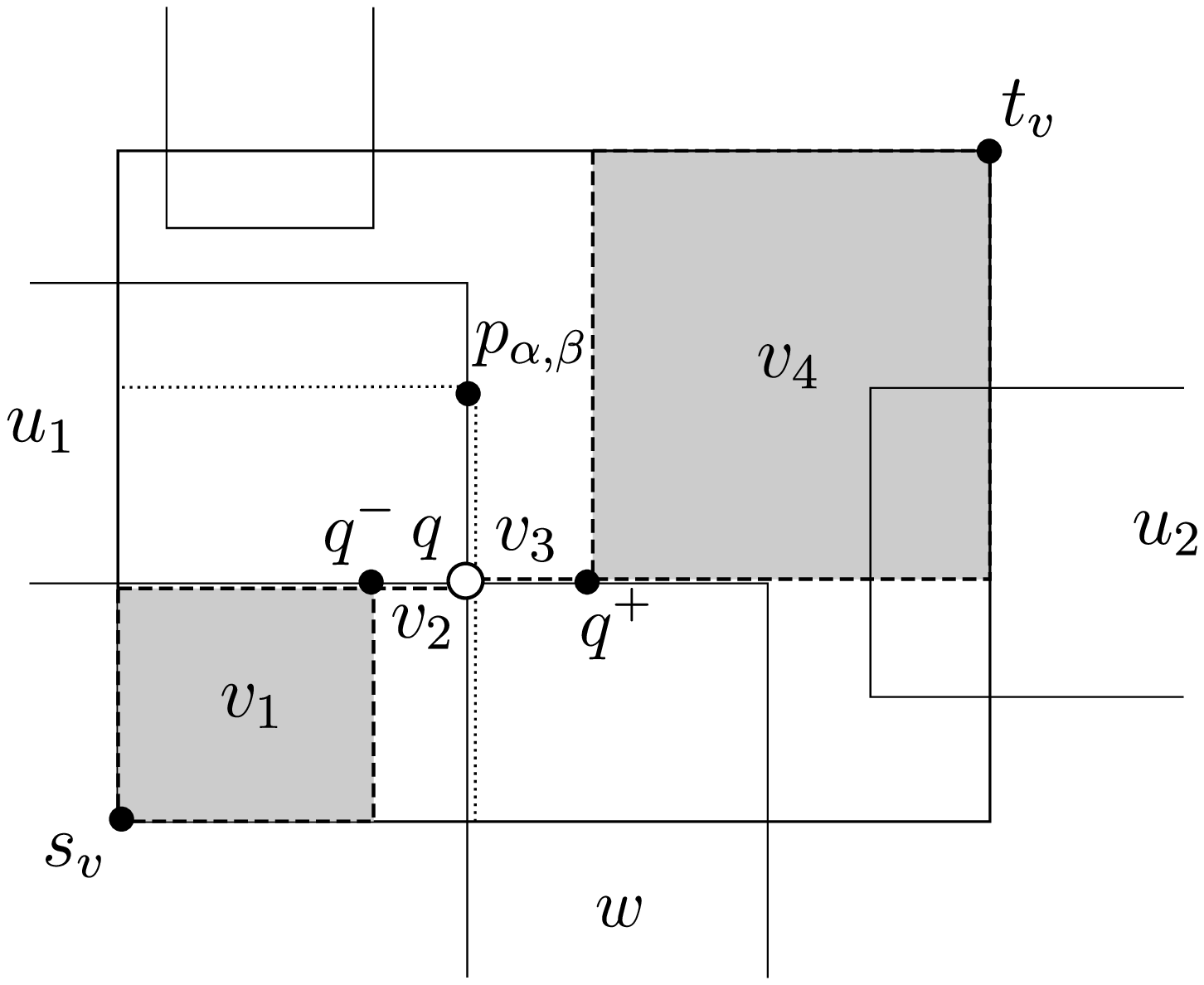}
        \subcaption{}
      \end{minipage}&
      \begin{minipage}{0.48\hsize}
        \centering
        \includegraphics[width=\hsize]{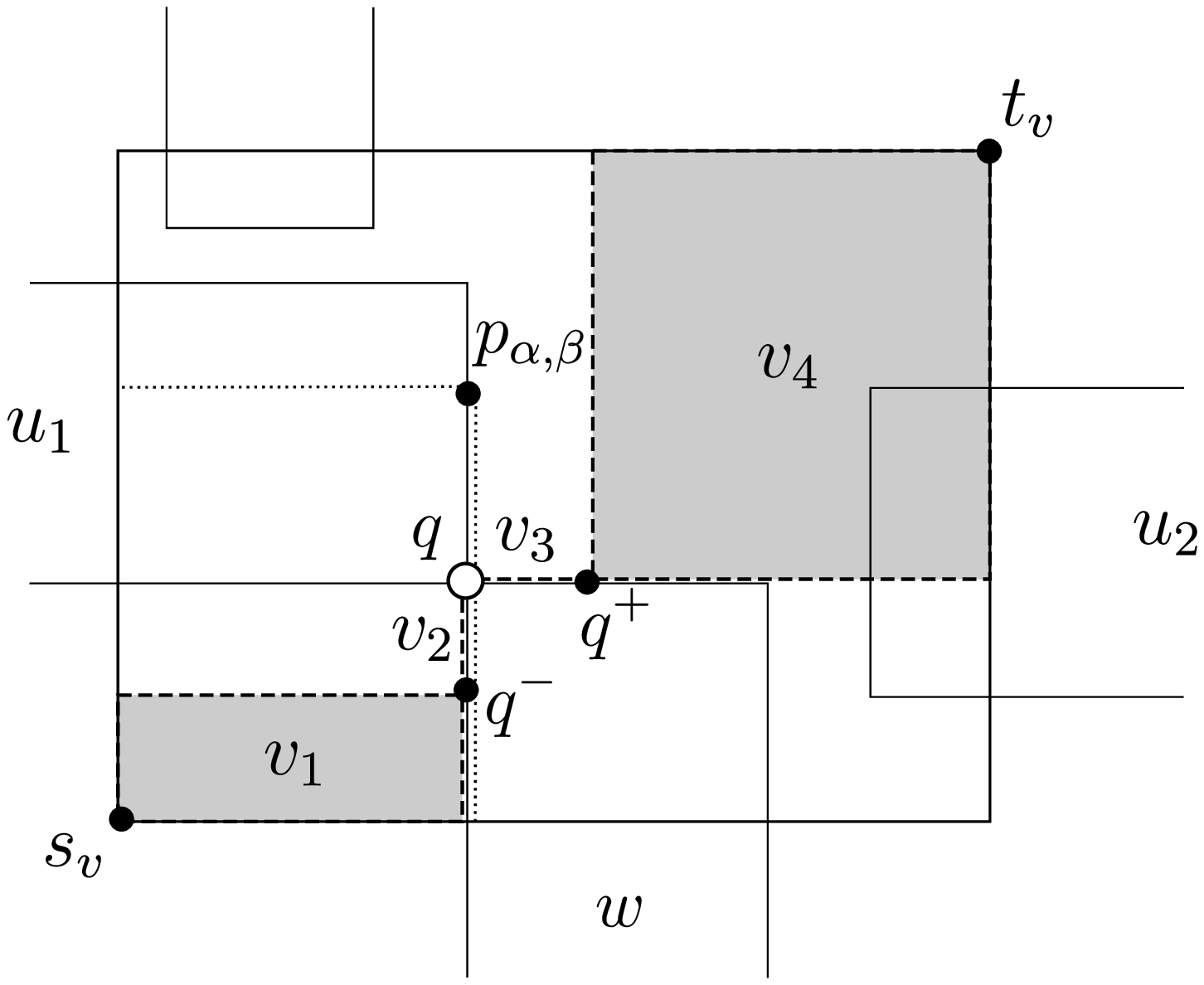}
        \subcaption{}
      \end{minipage}
    \end{tabular}
    \caption{An uneasy situation when $B(u_1)$ and $B(w)$ for some $w \in \Gamma_v \setminus \{u_1, u_2\}$ share their corners at $q$. (a) $\pi_v \in \Pi_\GMMN(v)$ goes through $q$. (b) $\pi_v \in \Pi_\GMMN(v)$ turns at $q$.}
    \label{fig:cycle2}
  \end{figure}
  
  Finally, we consider the situation when $q \in X^{\mathrm{vert}} - s_v$ and $(q^-, q, q^+) \in \mathcal{E}^{\mathrm{vert}}$.
  We have $i = \alpha$ by the definition of $X^\mathrm{vert}$.
  If $j < \beta$, then~(C1) does not hold because $B(v_1)$ does not intersect $B(u_2)$.
  Suppose that $j = \beta$, which means that $q = p_{\alpha, \beta} \in X^\mathrm{hor} \cap X^\mathrm{vert}$.
  Since the present $v_1$ is the same as that in the case of $(q^-, q, p_{\alpha,\beta+1}) \in \mathcal{E}^\mathrm{hor}$ if $\beta < b$, we have already proved that~(C1) does not hold in the above horizontal case (cf.~Figure~\ref{fig:cycle3}).
  It can be checked that the same proof is valid even if $\beta = b$.
  Thus~(C1) does not hold in any case.
  In addition, since $q_y \le q^+_y$ and $q_y \ne q^+_y$, we also have $B(v_4) \cap B(u_k) = \emptyset$ if $\alpha = h_k$ for $k = 1, 2$.
  Hence~(C2) is not the case either.
  
  Suppose that there exists $w \in \Gamma_v \setminus \{u_1, u_2\}$ that satisfies~(C3).
  As we have mentioned above, $B(v_4)$ does not intersect $B(u_2)$ if $\alpha = h_2$, which contradicts $\{v_4, u_2\} \in E(\IG[\tilde{\Gamma}])$.
  Suppose that $\alpha = h_1$.
  Let $\delta \in [b]$ such that $p_{\alpha, \delta}$ is the upper-left corner of $B(u_1) \cap B(v)$.
  If $j < \delta$, then $B(v_1)$ does not intersect $B(u_1)$, which contradicts $\{u_1, v_1\} \in E(\IG[\tilde{\Gamma}])$.
  If $j > \delta$, then $B(w)$ intersects $B(u_1)$; this contradicts the assumption that $\GMMN$ is triangle-free.
  Consider the remaining situation when $j = \delta$.
  If $q^- = p_{\alpha-1, j}$, then $B(w)$ intersects $B(u_1)$, a contradiction again.
  If $q^- = p_{\alpha, j-1}$, then $B(v_1)$ does not intersect $B(u_1)$.
  Therefore,~(C3) does not hold in any case, and we are done.
\end{proof}

From Claim~\ref{lem:reduction_forest}, we have successfully reduced the GMMN[Pseudotree] instance $P$ to $O(n)$ instances of GMMN[Tree] each of which has $n + 2$ pairs.

\begin{proposition}\label{prop:cycle}
\label{prop:cycle}
If \textup{GMMN[Tree]} can be solved in $T(n)$ time, then \textup{GMMN[Pseudotree]} can be solved in $O(n\cdot T(n+2))$ time.
\end{proposition}

By combining Proposition~\ref{prop:cycle} with Theorem~\ref{thm:speedup}, we obtain Corollary~\ref{cor:cycle}.

\section*{Acknowledgments}
We are grateful to the anonymous reviewers of the preliminary version~\cite{MOY} for their careful reading and giving helpful comments.
Most of this work was done when the first and third authors were with Osaka University, and the two authors appreciate the support from members of Department of Information and Physical Sciences.

\bibliographystyle{abbrv}
\bibliography{refs}

\clearpage
\appendix
\section{Faster Dynamic Programming on Tree Decompositions}\label{sec:tree_decomposition}
In this section, we prove the following theorem (cf.~Table~\ref{tab:previous_result}).

\begin{theorem}\label{thm:tree-decomposition}
There exists an $O(f(\tw, \Delta)\cdot n^{2\Delta(\tw + 1) + 1})$-time algorithm for the GMMN problem, where $\tw$ and $\Delta$ denote the treewidth and the maximum degree of the intersection graph $\IG[P]$ for the input $P$, respectively, and $f$ is a computable function.
\end{theorem}

\subsection{Treewidth and Nice Tree Decompositions}
We first review the concepts of tree decompositions and treewidth of graphs, and then define ``nice'' tree decompositions, which are useful to design a DP algorithm (cf.~\cite[Section~7.3]{parameterized-algorithms}).

\begin{definition}\label{def:tree-decomposition}
A \emph{tree decomposition} of an undirected graph $G$ is a pair $\mathcal{T} = (T, (X_t)_{t \in V(T)})$ of a tree $T$ and a tuple of subsets of $V(G)$ indexed by $V(T)$ such that the following three conditions hold:
\begin{enumerate}[label=(T\arabic{enumi}),labelindent=\parindent,leftmargin=*]
    \item $\bigcup_{t \in V(T)} X_t = V(G)$. \label{def:tree-decomposition:prop1}
    \item For every $\{u, v\} \in E(G)$, there exists $t \in V(T)$ such that $X_t$ contains both $u$ and $v$.\label{def:tree-decomposition:prop2}
    \item For every $u \in V(G)$, the set $T_u=\{t \in V(T) \mid u \in X_t\}$ is connected in $T$. \label{def:tree-decomposition:prop3}
\end{enumerate}
We call each $t \in V(T)$ a \emph{node} and each $X_t$ a \emph{bag}.

The \emph{width} of a tree decomposition is the maximum size of its bag minus one. 
The \emph{treewidth} of a graph $G$, which is denoted by $\tw(G)$ (or simply by $\tw$), is the minimum width of a tree decomposition of $G$.
\end{definition}

We choose an arbitrary node of a tree decomposition as a root, and define a nice tree decomposition as follows.

\begin{definition}\label{def:nice-tree-decomposition}
A rooted tree decomposition $\mathcal{T}=(T, (X_t)_{t\in V(T)})$ is said to be \emph{nice} if the following conditions are satisfied:
\begin{itemize}
    \item $X_r = \emptyset$ for the root node $r$,
    \item $X_l = \emptyset$ for every leaf node $l \in V(T)$, and 
    \item every non-leaf node of $T$ is one of the following three types:
    \begin{itemize}
        \item {\bf Introduce node:} a node $t$ having exactly one child $t'$ such that $X_t = X_{t'} \cup\{v\}$ for some vertex $v \notin X_{t'}$; we say that $v$ is \emph{introduced} at $t$.
        \item {\bf Forget node:} a node $t$ having exactly one child $t'$ such that $X_t = X_{t'} \setminus \{w\}$ for some vertex $w \in X_{t'}$; we say that $w$ is \emph{forgotten} at $t$.
        \item {\bf Join node:} a node $t$ having two children $t_1$ and $t_2$ such that $X_t = X_{t_1} = X_{t_2}$.
    \end{itemize}
\end{itemize}
\end{definition}

By the condition~\ref{def:tree-decomposition:prop3}, every vertex of $V(G)$ is forgotten only once, but may be introduced several times.
Given any tree decomposition, one can efficiently transform it as nice without increasing the width.

\begin{lemma}[{cf.~\cite[Lemma~7.4]{parameterized-algorithms}}]\label{lem:tree-decomposition:nice}
Any graph $G$ admits a nice tree decomposition of width at most $\tw(G)$. Moreover, given a tree decomposition $\mathcal{T}=(T, (X_t)_{t \in V(T)})$ of $G$ of width $k$, one can compute a nice tree decomposition of $G$ of width at most $k$ that has $O(k\cdot|V(G)|)$ nodes in $O(k^2 \cdot \max\left\{|V(T)|, |V(G)|\right\})$ time.
\end{lemma}

\subsection{Algorithm Outline}\label{sec:outline_tree_decomposition}
We first sketch the idea of Schnizler's algorithm for GMMN[Tree], and then extend it to our DP algorithm on nice tree decompositions.

Let $\GMMN$ be a GMMN[Tree] instance.
Suppose that we fix an arbitrary M-path $\pi_v^* \in \Pi_\GMMN(v)$ for some $v \in \GMMN$ and consider only feasible networks $N = (\pi_w)_{w \in \GMMN} \in \Feas(\GMMN)$ with $\pi_v = \pi_v^*$.
Then, the instance $\GMMN$ is intuitively divided into two independent parts $\GMMN_v - v$ and $\GMMN \setminus \GMMN_v$, where recall that $\GMMN_v$ denotes the vertex set of the subtree of $\IG[\GMMN]$ rooted at $v$.
In particular, if $\|N\|$ is minimized (subject to $\pi_v = \pi_v^*$), then the restriction $N[\GMMN_v] = (\pi_w)_{w \in \GMMN_v}$ also attains the minimum length subject to $\pi_v = \pi_v^*$ (which is true for the other side $(\GMMN \setminus \GMMN_v) \cup \{v\}$).

In addition, once we fix the in-out pairs $(s'_u, t'_u)$ of $\pi_u \in \Pi_\GMMN(u)$ for all neighbors $u \in \Gamma_v$, we can restrict the candidates for such $M$-paths $\pi_v^* \in \Pi_\GMMN(v)$ on the corresponding coarse grid $\mathcal{H}(\GMMN', v)$, where $\GMMN' = \{ (s'_u, t'_u) \mid u \in \Gamma_v \} \cup \{v\}$.
The number of candidates for $\GMMN'$ is at most $((4n)^2)^{\delta_v} \leq (16n)^{2\Delta}$ and the number of candidates for $\pi_v^* \in \Pi_{\GMMN'}(v)$ for each possible $\GMMN'$ is at most $\binom{4\delta_v + 4}{2\delta_v + 2} \leq 2^{4\Delta + 4}$, where recall that $\delta_v$ denotes the degree of $v$ in $\IG[\GMMN]$ and $\Delta$ is the maximum degree of $\IG[\GMMN]$.
Based on these observations, one can design a DP algorithm from the leaves to the root on $\IG[\GMMN]$ that computes minimum-length partial networks $N[\GMMN_v] = (\pi_w)_{w \in \GMMN_v}$ subject to $\pi_v = \pi_v^*$ for $O((cn)^{2\Delta})$ possible M-paths $\pi_v^*$ for each $v \in \GMMN$, where $c$ is some constant.

Let us turn to our DP algorithm.
Let $\GMMN$ be a GMMN instance and $\mathcal{T} = (T, (X_t)_{t \in V(T)})$ be a nice tree decomposition of the intersection graph $\IG[\GMMN]$ of width $\tw$.
As with the DP for GMMN[Tree] sketched above, we construct partial solutions from the leaves to the root of $\mathcal{T}$.
From Lemma~\ref{lem:tree-decomposition:nice}, we can assume that $T$ has $O(\tw \cdot n)$ nodes.

For $t \in V(T)$, let $\GMMN_t$ be the union of all the bags appearing in the subtree of $T$ rooted at $t$, including $X_t$.
Then, the following lemma analogously holds, which implies that among all the feasible solutions $N=(\pi_w)_{w \in \GMMN} \in \Feas(\GMMN)$ satisfying $N[X_t] = (\pi^*_w)_{w \in X_t} $ for some fixed $(\pi^*_w)_{w \in X_t}$, all the minimum-length solutions have exactly the same length in $\GMMN_t$.

\begin{lemma}
Let $N_1 \in \Feas(\GMMN)$ and $N_2\in \Feas(\GMMN)$ be two feasible solutions for $\GMMN$ such that $N_1[X_t] = N_2[X_t]$. 
If $\left\|N_1[\GMMN_t]\right\| < \left\|N_2[\GMMN_t]\right\|$, then $N_2$ is suboptimal, i.e., $N_2 \not\in \Opt(\GMMN)$. 
\end{lemma}

\begin{proof}
Let $N_2'$ be the network obtained from $N_2$ by replacing $N_2[\GMMN_t]$ with $N_1[\GMMN_t]$.
As $\mathcal{T}$ is a tree decomposition of the intersection graph $\IG[\GMMN]$, if we remove all the vertices in $X_t$ from $\IG[\GMMN]$, then $\GMMN_t \setminus X_t$ is disconnected from its complement in the remaining graph (cf.~the condition (T3) in Definition~\ref{def:tree-decomposition}).
Moreover, $N_1$ and $N_2$ have the same M-paths for $X_t$, and hence the network $N_2'$ is still an feasible solution for $\GMMN$. 
In addition, $\left\|N_1[\GMMN_t]\right\| < \left\|N_2[\GMMN_t]\right\|$ implies that $\|N_2'\| < \|N_2\|$, and we are done. 
\end{proof}

Based on this lemma, we define subproblems for possible solutions in $X_t$ as follows: given a GMMN instance $\GMMN$ and an M-path $\pi^*_v \in \Pi_\GMMN(v)$ for each $v \in X_t$, we are required to find a network $\hat{N} = (\hat{\pi}_w)_{w \in \GMMN} \in \Feas(\GMMN)$ such that $\hat{N}$ minimizes $\|\hat{N}[\GMMN_t]\|$ subject to $\hat{\pi}_v = \pi^*_v$ for all $v \in X_t$.
Formally, we define 
\begin{align}
    \twdp(t, (\pi_v^*)_{v \in X_t}) = \min\left\{ \left\|N[\GMMN_t] \right\| \bigm|
         N = (\pi_w)_{w \in \GMMN} \in \Feas(\GMMN),\ \pi_v = \pi_v^* \ (\forall v \in X_t) \right\}. \label{eq:dp_tree_decomposition:def}
\end{align}
If $t$ is a leaf, i.e., when $X_t = \emptyset$, then we write $(\pi_v^*)_{v \in X_t} = \epsilon$.
As with the tree case, it suffices to consider $O((cn)^{2\Delta})$ candidates for $\pi_v^* \in \Pi_\GMMN(v)$, and hence there exist $O((cn)^{2\Delta(\tw + 1)})$ candidates for $(\pi_v^*)_{v \in X_t}$ as $|X_t| \leq \tw + 1$.
We describe recursive formulae for filling up the DP table in the next section.

\subsection{Recursive Formula}
We separately discuss the four types of nodes in a nice tree decomposition (cf.~Definition~\ref{def:nice-tree-decomposition}).

\paragraph{Leaf node.} If $t$ is a leaf node, then $\twdp(t, \epsilon) = 0$ since $X_t = \emptyset$ and $\GMMN_t = \emptyset$. 

\paragraph{Introduce node.} If $t$ is an introduce node with the child $t'$ such that $X_t = X_{t'} \cup \{w\}$ for some $w \in X_{t'}$, then 
\begin{align}
    \twdp(t, (\pi_v^*)_{v \in X_t}) = \twdp(t', (\pi_v^*)_{v \in X_{t'}}) + d(s_w, t_w) - \left\|\pi^*_w \cap N^*\right\|, \label{eq:dp:introduce-node}
\end{align}
where we define $N^* = \bigcup_{v \in X_{t'}} \pi_v^*$ and the correctness of \eqref{eq:dp:introduce-node} is shown as follows.
If $w$ has a neighbor $u$ in $\GMMN_t \setminus X_{t'}$, then the edge $\{u, w\} \in E(\IG[\GMMN])$ cannot belong to any bag (because $u$ has already been forgotten in the subtree rooted at $t'$), which contradicts that $\mathcal{T}$ is a tree decomposition of $\IG[\GMMN]$ (cf.~the condition (T2) in Definition~\ref{def:tree-decomposition}).
Hence, it suffices to care the total length of segments shared by $\pi_w^*$ and $N^*$, which leads to the formula \eqref{eq:dp:introduce-node}.

\paragraph{Forget node.} If $t$ is a forget node with the child $t'$ such that $X_t = X_{t'} \setminus \{u\}$ for some $u \in X_{t'}$, then
\begin{align}
    \twdp(t, (\pi_v^*)_{v \in X_{t}}) = \min\left\{ \twdp(t', (\pi_v)_{v \in X_{t'}}) \mid \pi_v = \pi^*_v \ (\forall v \in X_{t'}\setminus \{u\})\right\},\label{eq:dp:forget-node}
\end{align} 
whose correctness is shown as follows.
By definition, we have $\GMMN_t = \GMMN_{t'}$ and any network $N = (\pi_w)_{w \in \GMMN} \in \Feas(\GMMN)$ with $\pi_v = \pi_v^*$ $(\forall v \in X_{t'})$ satisfies $\pi_v = \pi_v^*$ $(\forall v \in X_t \subseteq X_{t'})$, and hence $\twdp(t, (\pi_v^*)_{v \in X_{t}}) \le \twdp(t', (\pi^*_v)_{v \in X_{t'}})$.
On the other hand, if we take a network $N = (\pi_w)_{w \in \GMMN} \in \Feas(\GMMN)$ attaining $\twdp(t', (\pi_v^*)_{v \in X_{t'}})$, then we have $\twdp(t', (\pi_v^*)_{v \in X_{t'}}) = \twdp(t, (\pi_v)_{v \in X_{t}})$ and $\pi_v = \pi_v^*$ for all $v \in X_{t'}$.
Thus, we see that the formula \eqref{eq:dp:forget-node} holds.

\paragraph{Join node.} If $t$ is a join node with two children $t_1$ and $t_2$ such that $X_t = X_{t_1} = X_{t_2}$, then
\begin{align}
    \twdp(t,(\pi_v^*)_{v \in X_t}) = \twdp(t_1, (\pi_v^*)_{v \in X_t}) + \twdp(t_2, (\pi_v^*)_{v \in X_t}) - \|N^*\|,\label{eq:dp:join-node}
\end{align}
where we define $N^* = \bigcup_{v \in X_t} \pi_v^*$ and the correctness of \eqref{eq:dp:join-node} is shown as follows.
Let $N_1 = (\pi^{(1)}_w)_{w \in \GMMN} \in \Feas(\GMMN)$ and $N_2 = (\pi_w^{(2)})_{w \in \GMMN} \in \Feas(\GMMN)$ be networks attaining $\twdp(t_1, (\pi_v^*)_{v \in X_t})$ and $\twdp(t_2, (\pi_v^*)_{v \in X_t})$, respectively.
Since two pairs $u \in \GMMN_{t_1} \setminus X_{t}$ and $w \in \GMMN_{t_2} \setminus X_{t}$ cannot be adjacent (otherwise, the edge $\{u, w\} \in E(\IG[\GMMN])$ cannot belong to any bag, a contradiction), the restrictions $N_1[\GMMN_{t_1} \setminus X_t]$ and $N_2[\GMMN_{t_2} \setminus X_t]$ do not share any nontrivial segment.
In addition, as $\GMMN_t = \GMMN_{t_1} \cup \GMMN_{t_2}$ and $\pi^{(1)}_v = \pi^{(2)}_v = \pi_v^*$ for all $v \in X_{t} = X_{t_1} = X_{t_2}$, we obtain 
\begin{align}
    \twdp(t, (\pi_v^*)_{v \in X_t}) &= \|N_1\| + \|N_2\| - \|N^*\|= \twdp(t_1, (\pi_v^*)_{v \in X_t}) + \twdp(t_2, (\pi_v^*)_{v \in X_t}) - \|N^*\|.
\end{align}

\subsection{Computational Time Analysis}
In this section, we show that the whole algorithm runs in $O(f(\tw, \Delta) \cdot n^{2\Delta(\tw + 1) + 1})$ time, which completes the proof of Theorem~\ref{thm:tree-decomposition}.
Recall that $|V(T)| = O(\tw \cdot n)$ and the size of the DP table $\twdp(t, \cdot)$ is bounded by $O((cn)^{2\Delta(\tw + 1)})$ for each node $t \in V(T)$ (cf.~Section~\ref{sec:outline_tree_decomposition}).

For each forget node $t$, we just reduce the table $\twdp(t', \cdot)$ of size $O((cn)^{2\Delta(\tw + 1)})$ for the child $t'$ by taking the minimum according to \eqref{eq:dp:forget-node}, which requires $O((cn)^{2\Delta(\tw + 1)})$ time in total.

Recall that, when we consider candidates for $\pi_v^* \in \Pi_\GMMN(v)$, we restrict them on a coarse grid $\mathcal{H}(P', v)$ with $|P'| \leq \delta_v + 1 \leq \Delta + 1$ (cf.~Section~\ref{sec:outline_tree_decomposition}).
Hence, for each introduce or join node $t$, by \eqref{eq:dp:introduce-node} or \eqref{eq:dp:join-node}, respectively, one can compute $\twdp(t, (\pi_v^*)_{v \in X_t})$ in time depending only on $\Delta$ and $\tw$.

Thus, one can fill up each table $\twdp(t, \cdot)$ in $O(f(\tw, \Delta) \cdot n^{2\Delta(\tw + 1)})$ time, where $c^{2\Delta(\tw + 1)}$ is also included in $f(\tw, \Delta)$.
This concludes that the overall computational time is bounded by $O(f(\tw, \Delta) \cdot n^{2\Delta(\tw + 1) + 1})$.

We remark that, in the running time of our algorithm, the exponent of $n$ still contains both the treewidth $\tw$ and the maximum degree $\Delta$ of the intersection graph.
It remains open whether the GMMN problem is fixed parameter tractable (FPT) with respect to such parameters or not.

\section{Approximation Ratio Based on Chromatic Number}\label{sec:coloring}
In this section, we give a simple observation based on graph coloring.

\begin{table}[t]
    \centering
    \caption{Current best approximation ratios classified by the class of intersection graphs, whose treewidth and maximum degree are denoted by $\tw$ and $\Delta$, respectively.}
    \small
    \begin{tabular}{ccc}\hline
         Class & Approximation Ratio & Strategy\\\hline\hline
         \multirow{2}{*}{General} & $(6 + \epsilon)\log n $ & Divide-and-conquer \cite{10.1007/s00453-017-0298-0}\\\cline{2-3}
         & $O(\mathcal{D})$ ($\mathcal{D}=\Theta(\log n)$) & Divide-and-conquer \cite{funke2014generalized}\\\hline
        \multirow{3}{*}{\begin{tabular}{c}Complete Graphs\\$(\tw = \Delta = n - 1)$\end{tabular}} & $2+\epsilon$ & \begin{tabular}{c}
            Using PTASes for RSA \cite{zachariasen2000approximation,lu2000polynomial} \\
            (cf.~\cite[Lemma~3]{10.1007/s00453-017-0298-0})
         \end{tabular}\\ \cline{2-3}
         & $4$ & \begin{tabular}{c}
            Using a $2$-approximation algorithm \\for RSA \cite{rao1992rectilinear} 
            (cf.~\cite[Lemma~3]{10.1007/s00453-017-0298-0})
         \end{tabular}\\\hline \hline
         \multirow{2}{*}{General} & $ \Delta$ & 
              Coloring (Corollary~\ref{coro:degree-bounded})\\\cline{2-3}
              & $\tw + 1$ & Coloring (Corollary~\ref{coro:tw-bounded})\\\hline
         Planar Graphs & 4 & Coloring (Corollary~\ref{coro:planar})\\ \hline
     \end{tabular}
    \label{tab:previous_result:approx}
\end{table}

\begin{proposition}
\label{prop:coloring-IG}
Let $\GMMN$ be a GMMN instance and $N^*\in \Opt(\GMMN)$ be an optimal solution for $\GMMN$. If the intersection graph $\IG[\GMMN]$ of $\GMMN$ is $k$-colorable, for every $N \in \Feas(\GMMN)$, the total length of $N$ is at most $k$ times of the total length of $N^*$.
\end{proposition}
\begin{proof}
Let $\GMMN$ be a GMMN instance such that $\IG[\GMMN]$ is $k$-colorable, i.e., there exists a $k$-partition $\{\GMMN_1, \GMMN_2, \dots, \GMMN_k\}$ of $\GMMN$ such that every $\GMMN_i$ is an independent set in $\IG[\GMMN]$.
Then, for each $i \in [k]$, the total length of an optimal solution $N_i^* = (\pi_w)_{w \in \GMMN_i} \in \Opt(\GMMN_i)$ for the GMMN subinstance $\GMMN_i$ is equal to the sum of the Manhattan distances, i.e., 
\begin{equation}
    \left\|N_i^*\right\| = \sum_{w = (s, t) \in \GMMN_i} d(s, t).
\end{equation}
For every $i \in [k]$, the optimal solution $N^*$ also contains an M-path for every pair in $\GMMN_i$ since $\GMMN_i \subseteq \GMMN$, and hence $\|N_i^*\| \le \|N^*\|$.
Since any feasible solution $N \in \Feas(\GMMN)$ is written as $\bigcup_{i \in [k]} N_i^*$ for some $N_i^* \in \Opt(\GMMN_i)$ $(i \in [k])$ by definition, we have 
\begin{equation}
\|N\| = \left\|\bigcup_{i \in [k]} N_i^* \right\| \leq \sum_{i \in [k]} \|N_i\| \le k \cdot \|N^*\|,
\end{equation} 
and we are done.
\end{proof}

From Lemma~\ref{prop:coloring-IG}, we immediately obtain the following corollaries. For complete graphs and odd cycles, obviously, one needs $\Delta + 1$ colors, where $\Delta$ is the maximum degree. However, all other connected graphs are $\Delta$-colorable~\cite{brooks_1941}. Since the GMMN problem whose intersection graph is a complete graph and a cycle admits an $O(1)$-approximation algorithm and a polynomial-time (exact) algorithm, respectively, we focus on approximation ratio for other cases.
\begin{corollary}\label{coro:degree-bounded}
Let $\GMMN$ be a GMMN instance whose intersection graph has maximum degree at most $\Delta$, and is neither a complete graph nor an odd cycle. Let $N^*\in \Opt(\GMMN)$ be an optimal solution for $\GMMN$. Then for any feasible solution $N \in \Feas(\GMMN)$, we have $\|N\| \le \Delta \cdot \|N^*\|$.
\end{corollary}
It is easy to check that a graph of treewidth at most $\tw$ is $(\tw+1)$-colorable. 
\begin{corollary}\label{coro:tw-bounded}
Let $\GMMN$ be a GMMN instance whose intersection graph is of treewidth at most $\tw$. Let $N^*\in \Opt(\GMMN)$ be an optimal solution for $\GMMN$. Then for any feasible solution $N \in \Feas(\GMMN)$, we have $\|N\| \le (\tw+1)\cdot \|N^*\|$.
\end{corollary}
Every planar graph is known to be $4$-colorable~\cite{appel1977part1,appel1977part2,robertson1997}.
\begin{corollary}\label{coro:planar}
Let $\GMMN$ be a GMMN instance whose intersection graph is planar. Let $N^*\in \Opt(\GMMN)$ be an optimal solution for $\GMMN$. Then for any feasible solution $N \in \Feas(\GMMN)$, we have $\|N\| \le 4 \cdot \|N^*\|$.
\end{corollary}

\end{document}